\documentclass[11pt]{article}
\usepackage[utf8]{inputenc}
\usepackage{amsmath, amsthm, amssymb, mathtools}
\usepackage{amsfonts}
\usepackage{amssymb}
\usepackage{bbm}
\usepackage{graphicx}
\usepackage{caption}
\usepackage{fullpage}
\usepackage{tikz}
\usepackage{microtype}
\usepackage{thmtools,thm-restate}
\usepackage[linesnumbered,ruled,procnumbered]{algorithm2e}
\usepackage[numbers,comma,sort&compress]{natbib}
\usepackage{hyperref}
\usepackage[capitalize]{cleveref}

\hypersetup{
    colorlinks=true,
    linkcolor=blue,
    filecolor=cyan,
    urlcolor=magenta,
}

\newtheorem{theorem}{Theorem}
\newtheorem{lemma}{Lemma}
\newtheorem{claim}[lemma]{Claim}
\newtheorem{corollary}[lemma]{Corollary}
\newtheorem{remark}{Remark}

\newtheorem{definition}[theorem]{Definition}

\newcommand{\norm}[1]{\left\|#1\right\|}
\newcommand{\pnorm}[2]{\bigl\|#2\bigr\|_#1}

\newcommand{\normalpnorm}[2]{\|#2\|_#1}
\newcommand{\Tr}{\mathrm{Tr}}
\newcommand{\rank}{\mathrm{rank}}
\newcommand{\polylog}{\mathrm{polylog}}
\newcommand{\poly}{\mathrm{poly}}

\newcommand{\sspan}{\mathrm{span}}
\newcommand{\rows}{\mathrm{rows}}

\newcommand{\bbr}{\mathbb{R}}
\newcommand{\bbc}{\mathbb{C}}
\newcommand{\Var}{\mathrm{Var}}

\newcommand{\C}{\mathbb{C}}
\newcommand{\E}{\mathbb{E}}

\SetAlgoProcName{Procedure}{Procedure}
\crefname{claim}{Claim}{Claims}
\crefname{procedure}{Procedure}{Procedures}
\crefname{equation}{Eq.}{Eqns.}
\crefname{algocfline}{Line}{Lines}

\newcommand{\acc}[1]{\hyperref[sample:#1]{Access~\ref*{sample:#1}}}

\def\>{\rangle}
\def\<{\langle}

\renewcommand{\emptyset}{\varnothing}
\newcommand{\range}[1]{[#1]}

\newcommand{\nc}{\newcommand}
\nc{\rnc}{\renewcommand}

\newcommand{\nn}{\nonumber}
\newcommand{\bi}{\begin{itemize}}
\newcommand{\ei}{\end{itemize}}
\newcommand{\bn}{\begin{enumerate}}
\newcommand{\en}{\end{enumerate}}
\def\beas#1\eeas{\begin{eqnarray*}#1\end{eqnarray*}}
\def\ba#1\ea{\begin{align}#1\end{align}}

\def\nn{\nonumber}

\def\L{\left}
\def\R{\right}

\def\eps{\epsilon}

\def\benum{\begin{enumerate}}
\def\eenum{\end{enumerate}}
\def\bit{\begin{itemize}}
\def\eit{\end{itemize}}
\def\bdesc{\begin{description}}
\def\edesc{\end{description}}

\nc\qand{\qquad\text{and}\qquad}
\nc\mnb[1]{\medskip\noindent{\bf #1}}
\nc\mn{\medskip\noindent}

\nc{\nnl}{\nn \\ &}  
\nc{\fot}{\frac{1}{2}} 
\nc{\oo}[1]{\frac{1}{#1}} 
\newcommand{\ben}{\begin{enumerate}}
\newcommand{\een}{\end{enumerate}}
\nc{\mc}{\mathcal}

\newcommand{\hd}[1]{\vspace{3mm} \noindent\textbf{#1}}

\nc{\onenorm}[1]{\L\| #1 \R\|_1}

\nc{\Ra}{\Rightarrow}
\nc{\zo}{\{0,1\}}
\nc{\normF}{\pnorm{F}}
\nc{\sumnormF}[1]{\sum_\ell \normF{#1_\ell}^2}
\nc{\sumt}{\sum_{t \in T}}
\nc{\ct}{^\dag}

\newcommand{\ev}{\frac{\epsilon}{300r^2(\tau+1)}}
\newcommand{\ef}{\frac{\epsilon}{300r^2}}
\newcommand{\es}{\frac{\epsilon}{400r^2}}

\begin{document}
\begin{titlepage}
\clearpage

\title{Quantum-inspired sublinear algorithm for solving low-rank semidefinite programming}

\author{Nai-Hui Chia\thanks{Department of Computer Science, University of Texas at Austin. Email: \{nai,linhh,chunhao\}@cs.utexas.edu} \and Tongyang Li\thanks{Department of Computer Science, Institute for Advanced Computer Studies, and Joint Center for Quantum Information and Computer Science, University of Maryland. Email: tongyang@cs.umd.edu} \and Han-Hsuan Lin$^{*}$ \and Chunhao Wang$^{*}$}

\date{\empty}
\maketitle
\thispagestyle{empty}
\begin{abstract}
Semidefinite programming (SDP) is a central topic in mathematical optimization with extensive studies on its efficient solvers. In this paper, we present a proof-of-principle \emph{sublinear-time} algorithm for solving SDPs with low-rank constraints; specifically, given an SDP with $m$ constraint matrices, each of dimension $n$ and rank $r$, our algorithm can compute any entry and efficient descriptions of the spectral decomposition of the solution matrix. The algorithm runs in time $O(m\cdot\poly(\log n,r,1/\varepsilon))$ given access to a sampling-based low-overhead data structure for the constraint matrices, where $\varepsilon$ is the precision of the solution. In addition, we apply our algorithm to a quantum state learning task as an application.

Technically, our approach aligns with 1) SDP solvers based on the matrix multiplicative weight (MMW) framework by Arora and Kale [TOC '12]; 2) sampling-based dequantizing framework pioneered by Tang [STOC '19]. In order to compute the matrix exponential required in the MMW framework, we introduce two new techniques that may be of independent interest:
\begin{itemize}
\item Weighted sampling: assuming sampling access to each individual constraint matrix $A_{1},\ldots,A_{\tau}$, we propose a procedure that gives a good approximation of $A=A_{1}+\cdots+A_{\tau}$.
\item Symmetric approximation: we propose a sampling procedure that gives the \emph{spectral decomposition} of a low-rank Hermitian matrix $A$. To the best of our knowledge, this is the first sampling-based algorithm for spectral decomposition, as previous works only give singular values and vectors.
\end{itemize}
\end{abstract}

\end{titlepage}


\newpage

\section{Introduction}
Semidefinite programming (SDP) is a central topic in the studies of mathematical optimization and theoretical computer science, with a wide range of applications including algorithm design, machine learning, operations research, etc. The importance of SDP comes from both its generality that contains the better-known linear programming (LP) and the fact that it admits polynomial-time solvers. Mathematically, an SDP is defined as follows:
\begin{align}
\max\quad&\Tr[CX] \label{eqn:SDP-OPT} \\
\text{s.t.}\quad&\Tr[A_i X]\leq b_{i}\quad\forall\,i\in\range{m}; \label{eqn:SDP-trace} \\
&X\succeq 0, \label{eqn:SDP-X}
\end{align}
where $m$ is the number of constraints, $A_{1},\ldots,A_{m},C$ are $n\times n$ Hermitian matrices, and $b_{1},\ldots,b_{m}\in\bbr$; \cref{eqn:SDP-X} restricts the variable matrix $X$ to be \emph{positive semidefinite} (PSD), i.e., $X$ is an $n\times n$ Hermitian matrix with non-negative eigenvalues (more generally, $X\succeq Y$ means that $X-Y$ is a PSD matrix). An $\varepsilon$-approximate solution of this SDP is an $X^{*}$ that satisfies \cref{eqn:SDP-trace,eqn:SDP-X} while $\Tr[CX^*]\geq\textsf{OPT}-\varepsilon$ (\textsf{OPT} being the optimum of the SDP).

There is rich literature on solving SDPs. Ellipsoid method gave the first polynomial-time SDP solvers \cite{khachiyan1980polynomial,grotschel1981ellipsoid}, and the complexities of the SDP solvers had been subsequently improved by the interior-point method \cite{nesterov1992conic} and the cutting-plane method \cite{anstreicher2000volumetric,mitchell2003polynomial}; see also the survey paper \cite{vandenberghe1996semidefinite}. The current state-of-the-art SDP solver \cite{lee2015faster,jiang2020improved} runs in time $\tilde{O}(m(m^2+n^{\omega}+mn^{2})\poly(\log 1/\varepsilon))$, where $\omega<2.373$ is the exponent of matrix multiplication.\footnote{Throughout the paper, $\tilde{O}(f(\cdot))$ denotes $O\L(f(\cdot)\polylog(f(\cdot))\R)$.} On the other hand, if we tolerate polynomial dependence in $1/\varepsilon$, Arora and Kale~\cite{arora2007combinatorial} gave an SDP solver with better complexities in $m$ and $n$: $\tilde{O}(mn^{2}(R_{\text{p}}R_{\text{d}}/\varepsilon)^{4} + n^{2}(R_{\text{p}}R_{\text{d}}/\varepsilon)^7)$, where $R_{\text{p}}$, $R_{\text{d}}$ are given upper bounds on the $\ell_{1}$-norm of the optimal primal and dual solutions, respectively (see more details in~\cite{vanApeldoorn2017quantum}). This is subsequently improved to $\tilde{O}(m/\varepsilon^{2} + n^{2}/\varepsilon^{2.5})$ by Garber and Hazan \cite{GH11,GH12} when $R_{\text{p}}, R_{\text{d}}=1$ and $b_{i}=0$ in \cref{eqn:SDP-trace} for all $i\in[m]$; as a complement, \cite{GH11} also established a lower bound $\Omega(m/\varepsilon^{2}+n^{2}/\varepsilon^{2})$ under the same assumption.

The SDP solvers mentioned above all use the standard entry-wise access to matrices $A_{1},\ldots,A_{m}$, and $C$. In contrast, a common methodology in algorithm design is to assume a certain natural \emph{preprocessed data structure} such that the problem can be solved in \emph{sublinear} time, perhaps even in \emph{poly-logarithmic} time, given queries to the preprocessed data structure (e.g., see the examples discussed in \cref{sec:related-works}). Such methodology is extensively exploited in \emph{quantum algorithms}, where we are given a unitary oracle to access entries of matrices in \emph{superposition}, a fundamental feature in quantum mechanics and the essence of quantum speedups. In particular, quantum SDP solvers in the case that matrices are sparse have been studied in \cite{brandao2016quantum,vanApeldoorn2017quantum,brandao2018SDP,vanApeldoorn2018SDP} and culminate in a quantum algorithm that runs in time $\tilde{O}\big((\sqrt{m}+\sqrt{n}R_{\text{p}}R_{\text{d}}/\varepsilon)n(R_{\text{p}}R_{\text{d}}/\varepsilon)^4\big)$ \cite{vanApeldoorn2018SDP}, which achieve \emph{polynomial} speedup comparing to existing classical algorithms in $m$ and $n$. Based on an oracle that can prepare the quantum state corresponding to the positive semidefinite part of Hermitian matrices,\footnote{Ref.~\cite{vanApeldoorn2018SDP} called this the quantum state input model. Their complexity is expressed in terms of a parameter $B$, which is basically the trace norm of the constraint and cost matrices, which then is basically the rank for matrices with spectral norm 1.} quantum \emph{exponential} speedup in $n$ has been achieved for matrices with rank $r=\poly(\log n)$  by~\cite{brandao2018SDP,vanApeldoorn2018SDP}, whose algorithms run in time $\tilde{O}(\sqrt{m})\cdot\poly(\log m,\log n,r,1/\epsilon)$. Considering this, SDP was previously believed to be a strong candidate for exponential quantum speedups in the low-rank setting (see e.g.~\cite{preskill2018NISQ}).

Mutually inspired by both classical and quantum SDP solvers, and the series of ``dequantization'' results~\cite{tang2018quantum2,chia2018quantum} lead by Tang's breakthrough result~\cite{tang2018quantum}, in this work we strive to match the exponential speedup of ~\cite{brandao2018SDP,vanApeldoorn2018SDP} with a classical algorithm, with the same low-rank requirement on constraint and cost matrices.

\subsection{Main results}
We show that when the constraint and cost matrices are low-rank, with a low-overhead data structure that supports the following sampling access, there exists a classical algorithm whose runtime is logarithmic in the dimension $n$ of the matrices.
\begin{definition}[Sampling and query  access]\label{defn:sampling-informal}
  Let $M \in \bbc^{n \times n}$. Denote $\norm{\cdot}$ to be the $\ell_{2}$ norm and $\norm{\cdot}_F$ to be the Frobenius norm. We say that we have the \emph{sampling access} to $M$ if we can
  \vspace{3mm}
  \begin{enumerate}
    \item\label{sample:row} sample a row index $i \in [n]$ of $M$ where the probability of row $i$ being chosen is
      \begin{align*}
        \frac{\norm{M(i, \cdot)}^2}{\pnorm{F}{M}^2};
      \end{align*}
    \item\label{sample:element} for all $i \in [n]$, sample an index $j \in [n]$ where the probability of $j$ being chosen is
      \begin{align*}
        \frac{|M(i, j)|^2}{\norm{M(i, \cdot)}^2};
      \end{align*}
      \item query the entry $M(i,j)$ for any $i,j\in[n]$; and
      \item evaluate norms of $\|M\|_F$ and $\|M(i,\cdot)\|$ for $i\in [n]$,
  \end{enumerate}
  \vspace{3mm}
  with time complexity $O(\poly(\log n))$ for each sampling and norm access.
\end{definition}
A low-overhead data structure that allows for this sampling access is shown in \cref{sec:datastruct}. Our main result is as follows.
\begin{theorem}[informal; see \cref{thm:main,alg:fea_testing,thm:SDP-master}]\label{thm:main-informal}
Let $C,A_1,\dots,A_m\in \mathbb{C}^{n\times n}$ be an SDP instance as in \crefrange{eqn:SDP-OPT}{eqn:SDP-X}. Suppose $\rank(C),\max_{i\in\range{n}}\rank(A_{i})\leq r$. Given sampling access to $A_{1},\ldots,A_{m},C$ in \cref{defn:sampling-informal}, there is an algorithm that gives any specific entry of an $\varepsilon$-approximate solution of the SDP with probability at least $2/3$; the algorithm runs in time $O(m\cdot\poly(\log n,r,R_{\text{p}}R_{\text{d}}/\varepsilon))$, where $R_{\text{p}}$, $R_{\text{d}}$ are given upper bounds on the $\ell_{1}$-norm of the optimal primal and dual solutions.
\end{theorem}

Comparing our results to existing classical randomized algorithms for solving SDP (e.g.,~\cite{arora2007combinatorial, GH11, GH12}), our algorithm outperforms existing classical SDP solvers given sampling access to the constraint matrices (which can be realized with a low-overhead data structure). Specifically, the running time of our algorithm is $O(m\cdot\poly(\log n,r,R_{\text{p}}R_{\text{d}}/\varepsilon))$ according to \cref{thm:main-informal}, which achieves exponential speedup in terms of $n$ with the data structure given in \cref{thm:ds}. It is worth noting that there are other ways to implement the sampling and query access. For example, Drineas, Kannan, and Mahoney~\cite[Lemma 2]{DKM06} showed that the sampling access in \cref{defn:sampling-informal} can be achieved with poly-logarithmic space if the matrix elements are streamed. Therefore, \cref{thm:main-informal} also implies that there exists a one-pass poly-logarithmic space algorithm for low-rank SDP in the data-streaming model.

Compared to quantum algorithms, our algorithm has comparable running time. It is because existing quantum SDP solvers that achieve exponential speed up in terms of $n$ all basically have polynomial dependence on the rank $r$~\cite{brandao2018SDP,vanApeldoorn2018SDP}, so they also have $\poly(\log n, r)$ complexity. It is worth noting that the quantum SDP solvers require additional assumptions on the way the matrices are given.  Furthermore, we give query access to the solution matrix which was not achieved by existing quantum SDP solvers, as they only give sampling access to the solution matrix. In this regard, it is easy to obtain the sampling access of the solution matrix from our algorithm by extending the rejection sampling techniques of~\cite{tang2018quantum} as pointed out by Tang\footnote{Personal communication.}.

Our result aligns with the studies of sampling-based algorithms for solving linear algebraic problems. In particular, \cite{FKV04} gave low-rank approximations of a matrix in \emph{poly-logarithmic} time with sampling access to the matrix as in \cref{defn:sampling-informal}. Recently, Tang extended the idea of \cite{FKV04} to give a poly-logarithmic time algorithm for solving recommendation systems \cite{tang2018quantum}. Subsequently, still under the same sampling assumption, Tang~\cite{tang2018quantum2} sketched poly-logarithmic algorithms for principal component analysis and clustering assignments, and two follow-up papers~\cite{GLT18,chia2018quantum} gave poly-logarithmic algorithms for solving low-rank linear systems. All these sampling-based sublinear algorithms directly exploit the sampling approach in \cite{FKV04} (see \cref{sec:techniques} for details); to solve SDPs, we derive an augmented sampling toolbox which includes two novel techniques: \emph{weighted sampling} and \emph{symmetric approximation}.

As a corollary, our SDP solver can be applied to learning quantum states\footnote{A quantum state $\rho$ is a PSD matrix with trace one.} efficiently. A particular task of learning quantum states is \emph{shadow tomography} \cite{aaronson2017quantum}, where we are asked to find a description of an unknown quantum state $\rho$ such that we can  approximate $\Tr[\rho E_{i}]$ up to error $\eps$ for a specific collection of Hermitian matrices $E_1, \ldots, E_m$ where $0\preceq E_i \preceq I$ and $E_i\in\C^{n\times n}$ for all $i\in\range{m}$ (such $E_{i}$ can also be viewed as a measurement operator in a two-outcome POVM in quantum computing). Mathematically, shadow tomography can be formulated as the following SDP feasibility problem:
\begin{align}
\text{Find $\sigma$ such that}\qquad |\Tr[\sigma E_{i}]-\Tr[\rho E_{i}]|&\leq\epsilon\quad\forall\,i\in\range{m}; \label{eqn:SDP-shadow-1} \\
\sigma\succeq 0,\ \ \Tr[\sigma]&=1. \label{eqn:SDP-shadow-2}
\end{align}
Under a quantum model proposed by \cite{brandao2018SDP} where $\rho,E_{1},\ldots,E_{m}$ are stored as quantum states, the state-of-the-art quantum algorithm \cite{vanApeldoorn2018SDP} solves the shadow tomography problem with time $O\big((\sqrt{m} + \min\{\sqrt{n}/\epsilon, r^{2.5}/\epsilon^{3.5}\})r/\epsilon^{4}\big)$ where $r=\max_{i\in\range{m}}\rank(E_{i})$; in other words, quantum algorithms achieve poly-logarithmic complexity in $n$ for low-rank shadow tomography. We observe the same phenomenon under our sampling-based model:
\begin{corollary}[informal; see \cref{cor:shadow}]\label{cor:shadow-informal}
Given sampling access of matrices $E_{1},\ldots,E_{m}  \in \bbc^{n\times n} $ as in \cref{defn:sampling-informal} and real numbers $\Tr[\rho E_{1}],\ldots,\Tr[\rho E_{m}]$, there is an algorithm that gives a succinct description as in \cref{remark:SDP-solution} and any entry of an $\epsilon$-approximate solution $\sigma$ of the shadow tomography problem defined as \cref{eqn:SDP-shadow-1,eqn:SDP-shadow-2} with probability at least $2/3$; the algorithm runs in time $O(m\cdot\poly(\log n,r,1/\epsilon))$.
\end{corollary}

\subsection{Techniques}\label{sec:techniques}
\hd{Matrix multiplicative weight method (MMW).}
We study a normalized SDP feasibility testing problem defined as follows:
\begin{definition}[Feasibility of SDP] \label{defn:feasibility}
Given an $\epsilon>0$, $m$ real numbers $a_{1},\ldots,a_{m}\in\bbr$, and Hermitian $n\times n$ matrices $A_{1},\ldots,A_{m}$ where $-I\preceq A_{i}\preceq I, \forall j\in\range{m}$, define $\mathcal{S}_{\epsilon}$ as the set of all $X$ such that
\begin{align}
\Tr[A_{i} X] &\leq a_{i}+\epsilon\quad\forall\,i\in\range{m}; \label{eqn:SDP-1} \\
X&\succeq 0; \label{eqn:SDP-2} \\
\Tr[X]&=1. \label{eqn:SDP-3}
\end{align}
For $\epsilon$-approximate feasibility testing of the SDP, we require that:
\begin{itemize}
\item If $\mathcal{S}_{\epsilon}=\emptyset$, output ``infeasible'';
\item If $\mathcal{S}_{0}\neq\emptyset$, output an $X\in\mathcal{S}_{\epsilon}$.\footnote{If $\mathcal{S}_{\epsilon}\neq\emptyset$ and $\mathcal{S}_{0}=\emptyset$, either output is acceptable.}
\end{itemize}
\end{definition}
It is well-known that one can use binary search to reduce $\varepsilon$-approximation of the SDP in \crefrange{eqn:SDP-OPT}{eqn:SDP-X} to $O(\log(R_{\text{p}}R_{\text{d}}/\varepsilon))$ calls of the feasibility problem in \cref{defn:feasibility} with $\epsilon=\varepsilon/(R_{\text{p}}R_{\text{d}})$.\footnote{\label{footnote:binary-search}For the normalized case $R_{\text{p}}R_{\text{d}}=1$, we first guess a candidate value $c_{1}=0$ for the objective function, and add that as a constraint $\Tr[CX]\geq c_{1}$ to the optimization problem. If this problem is feasible, the optimum is larger than $c_{1}$ and we accordingly take $c_{2}=c_{1}+\frac{1}{2}$; if this problem is infeasible, the optimum is smaller than $c_{1}$ and we accordingly take $c_{2}=c_{1}-\frac{1}{2}$; we proceed similarly for all $c_{i}$. As a result, we could solve the optimization problem with precision $\epsilon$ using $\lceil\log_{2}\frac{1}{\epsilon}\rceil$ calls to the feasibility problem in \cref{defn:feasibility}. For renormalization, it suffices to take $\epsilon=\varepsilon/(R_{\text{p}}R_{\text{d}})$. See also~\cite{brandao2018SDP}.} Therefore, in this paper we focus on solving feasibility testing of SDPs.

To solve the feasibility testing problem in \cref{defn:feasibility}, we follow the \emph{matrix multiplicative weight} (MMW) framework \cite{arora2012survey}. To be more specific, we follow the approach of regarding MMW as a zero-sum game with two players (see, e.g., \cite{Hazan,Wu10,gutoski2012parallel,LRS15,brandao2018SDP}), where the first player wants to provide a feasible $X \in \mathcal{S}_{\epsilon}$, and the second player wants to find any violation $j\in\range{m}$ of any proposed $X$, i.e., $\Tr[A_{j}X]>a_{j}+\epsilon$. At the $t^{\text{th}}$ round of the game, if the second player points out a violation $j_{t}$ for the current proposed solution $X_{t}$, the first player proposes a new solution
\begin{align}\label{eqn:Gibbs-defn}
X_{t+1}\leftarrow\exp[-(A_{j_{1}}+\cdots+A_{j_{t}})]
\end{align}
(up to normalization); such a solution by matrix exponentiation is formally named as a \emph{Gibbs state}. A feasible solution is actually an equilibrium point of the zero-sum game, which can also be approximated by the MMW method \cite{arora2012survey}; formal discussions are given in \cref{sec:MMW-SDP}.

\hd{Improved sampling tools.}
Before describing our improved sampling tools, let us give a brief review of the techniques introduced by \cite{FKV04}. The basic idea of \cite{FKV04} comes from the fact that a low-rank matrix $A$ can be well-approximated by sampling a small number of rows. More precisely, suppose that $A$ is an $n\times n$ matrix with rank $r$, where $n\gg r$. Because $n$ is large, it is preferable to obtain a ``description'' of $A$ without using $\poly(n)$ resources. If we have the sampling access to $A$ in the form of \cref{defn:sampling-informal}, we can sample rows from $A$ according to statement 1 of \cref{defn:sampling-informal}. Suppose $S$ is the $p\times n$ submatrix of $A$ formed by sampling $p=\poly(r)$ rows from $A$ with some normalization. It can be shown that $S^\dag S \approx A^\dag A$ in the Frobenius norm. Furthermore, we can apply the similar sampling techniques to sampling $p$ columns of $S$ with some normalization to form a $p \times p$ matrix $W$ such that $W W^\dag \approx S S^\dag$. Then the singular values and singular vectors of $W$, which are easy to compute because $p$ is small, together with the row indices that form $S$, can be viewed as a succinct description of some matrix $V \in \bbc^{n \times r}$ satisfying $A \approx AVV^{\dag}$, which gives a low-rank projection of $A$. In~\cite{chia2018quantum}, this method was extended to approximating the spectral decomposition of $AA^{\dag}$, i.e., calculating a small diagonal matrix $D$ and finding a succinct description of $V$  such that $VD^2V^{\dag} \approx A A^\dag$.

To implement the MMW framework, we need an approximate description of the matrix exponentiation $X_{t+1}:=\exp[-\sum_{\tau=1}^{t} A_{j_{\tau}}]$ in \cref{eqn:Gibbs-defn}. We achieve this in two steps. First, we get an approximate description of the spectral decomposition  of $A:=\sum_{\tau=1}^{t} A_{j_{\tau}}$ as $A \approx  \hat{V} \hat{D} \hat{V}^\dag$, where $\hat{V}$ is an $n\times r$ matrix and $\hat{D}$ is an $r\times r$ real diagonal matrix. Then, we approximate the matrix exponentiation $e^{-A}$ by $\hat{V} e^{-D} \hat{V}^\dag$.

There are two main technical difficulties that we overcome with new tools while following the above strategy. First, since $A$ changes dynamically throughout the MMW method, we cannot assume the sampling access to $A$; a more reasonable assumption is to have sampling access to each individual constraint matrix $A_{j_{\tau}}$, but it is hard to directly sample from $A$ by sampling from each individual $A_{j_{\tau}}.$\footnote{For example, assume we have $A=A_1+A_2$ such that $A_2=-A_1+\boldsymbol{\eps}$, where $\boldsymbol{\eps}$ is a matrix with small entries. In this case, $A_1$ and $A_2$ mostly cancel out each other and leave $A=\boldsymbol{\eps}$. Since $\boldsymbol{\eps}$ can be arbitrarily small compared to $A_1$ and $A_2$, it is hard to sample from $\boldsymbol{\eps}$ by sampling from $A_1$ and $A_2$.}
In \cref{sec:matrix_samp}, we sidestep this difficulty by devising the \emph{weighted sampling} procedure which gives a succinct description of a low-rank approximation of $A=\sum_{\tau} A_{j_\tau}$ by sampling each individual $A_{j_\tau}$. In other words, we cannot sample according to $A$, but we can still find a succinct description of a low-rank approximation of $A$.

Second, the original sampling procedure of \cite{FKV04} and the extension by \cite{chia2018quantum} give $V  D^2 V^\dag \approx A^\dag A$ instead of a spectral decomposition $\hat{V}\hat{D}\hat{V}^\dag \approx A$, even if $A$ is Hermitian. For our purpose of matrix exponentiation, singular value decomposition is problematic because the singular values ignore the signs of the eigenvalues; specifically, we get a large error if we approximate $e^{-A}$ by naively exponentiating the singular value decomposition: $e^{-A} \not\approx V e^{-D} V^\dag$. Note that this issue of missing negative signs is intrinsic to the tools in \cite{FKV04} because they are built upon the approximation $S^\dag S\approx A^\dag A$; Suppose that $A$ has the decomposition $A= U D V^\dag$, where $D$ is a diagonal matrix, and $U$ and $V$ are isometries. Then $A^\dag A= V D^\dag D V^\dag$, cancelling out any phase on $D$. We resolve this issue by a novel approximation procedure, \emph{symmetric approximation}. Symmetric approximation is based on the result $A \approx A V V^\dag$ shown by \cite{FKV04}. It then holds that $A \approx V (V^\dag A V) V^\dag $ because the symmetry of $A$ implies that $VV^\dag$ acts roughly as the identity on the image of $A$. Since $(V^\dag A V)$ is a small matrix of size $r\times r$, we can calculate it explicitly and diagonalize it, getting approximate eigenvalues of $A$. Together with the description of $V$, we get the desired description of the spectral decomposition of $A$. See \cref{sec:approx-eigen} for more details.\footnote{It might be illustrative to describe some of our failed attempts before achieving symmetric approximation. We tried to separate the exponential function into even and odd parts; unfortunately that decomposes $e^{-x}$ into $e^{-x} =\cosh{x}-\sinh{x}$. Even if $\cosh{x}$ and $\sinh{x}$ both have bounded relative error, because $e^{-x}$ could be small, it has  unbounded relative error. We also tried to obtain the eigenvectors of $A$ from $V$; this approach faces multiple difficulties, the most serious one being the ``fake degeneracy'' as shown by the following example. Suppose $A=\bigl(\begin{smallmatrix} 1 & 0\\ 0 & -1 \end{smallmatrix}\bigr)$. $A$ has two distinct eigenvectors. However, $A^\dag A= V D D V^\dag$ can be satisfied by taking $D=I$ together with any unitary $V$. In this case, $V$ does not give any information about the eigenvectors.}

\subsection{Related work}\label{sec:related-works}
As we have mentioned earlier, many SDP solvers use cutting-plane methods or interior-point methods with complexity $\poly(\log(1/\epsilon))$ but larger complexities in $m$ and $n$. In contrast, our SDP solver follows the MMW framework, and we briefly summarize such SDP solvers in existing literature. They mainly fall into two categories as follows.

First, MMW is adopted in solvers for \emph{positive} SDPs, i.e., $A_{1},\ldots,A_{m},C\succeq 0$. In this case, the power of MMW lies in its efficiency of having only $\tilde{O}(\poly(1/\epsilon))$ iterations (i.e., poly-logarithmic in $m,n$) and the fact that it admits \emph{width-independent} solvers whose complexities are independent of $R_{\text{p}}$ and $R_{\text{d}}$. Ref.~\cite{luby1993parallel} first gave a width-independent positive LP solver that runs in $O(\log^{2}(mn)/\epsilon^{4})$ iterations, and \cite{jain2011parallel} subsequently generalized this result to give the first width-independent positive SDP solver.
The state-of-the-art positive SDP solver was proposed by \cite{allen-zhu2016SDP} with only $O(\log^{2}(mn)/\epsilon^{3})$ iterations.

Second, as far as we know, the vast majority of \emph{quantum} SDP solvers use the MMW framework. The first quantum SDP solver was proposed by \cite{brandao2016quantum} with worst-case running time $\tilde{O}(\sqrt{mn}s^2(R_{\text{p}}R_{\text{d}}/\epsilon)^{32})$, where $s$ is the sparsity of input matrices, i.e., every row or column of $A_{1},\ldots,A_{m},C$ has at most $s$ nonzero elements. Subsequently, the quantum complexity of solving SDPs was improved by \cite{vanApeldoorn2017quantum,brandao2018SDP}, and the state-of-the-art quantum SDP solver runs in time $\tilde{O}\big((\sqrt{m}+\sqrt{n}R_{\text{p}}R_{\text{d}}/\eps)s(R_{\text{p}}R_{\text{d}}/\eps)^4\big)$ \cite{vanApeldoorn2018SDP}. This is optimal in the dependence of $m$ and $n$ because \cite{brandao2016quantum} proved a quantum lower bound of $\Omega(\sqrt{m}+\sqrt{n})$ for constant $R_{\text{p}},R_{\text{d}},s$, and $\epsilon$.

The authors, along with Gily\'{e}n and Tang~\cite{CGLLTW19}, further generalized the techniques to singular value transformation\footnote{Similar results were simultaneously obtained by Jethwani, Le Gall, and Singh~\cite{JLS19}.} and proposed a quantum-inspired framework to dequantize almost all known quantum machine learning algorithms with claimed exponential speedups. The low-rank SDP problem also fits in that framework and a better time complexity has been achieved in~\cite{CGLLTW19} due to a more efficient sampling method. However, the results in this paper remain its own interest: our techniques are specially crafted for approximating matrices in the form of $e^{-\epsilon A}$ and hence provide additional insights for this line of research.

\subsection{Open questions}
Our paper raises a couple of natural open questions for future work. For example:
\begin{itemize}
\item Can we give faster sampling-based algorithms for solving LPs? Note that recent breakthroughs by \cite{cohen2018solving,jiang2020faster} solve LPs with complexity $\tilde{O}(n^{\omega})$ where $\omega\approx 2.373$, significantly faster than the state-of-the-art SDP solver \cite{lee2015faster} with complexity $\tilde{O}(m(m^2+n^{\omega}+mn^{2}))$.

\item What is the empirical performance of our sampling-based method? Admittedly, the exponents of our poly-logarithmic factors are large; nevertheless, it is common that numerical experiments perform better than theoretical guarantees, and we wonder if this phenomenon can be observed when applying our method.
\end{itemize}

\hd{Organization.}
The rest of the paper is organized as follows. We formulate the sampling-based data structure and the SDP feasibility problem in \cref{sec:prelim}. Our two techniques, weighted sampling and symmetric approximation, are presented in \cref{sec:matrix_samp} and \cref{sec:approx-eigen}, respectively. Subsequently, we apply these techniques to estimate traces with respect to Gibbs states in \cref{sec:gibbs}, and give the proof of the main results in \cref{sec:proof-main}.


\section{Preliminaries}\label{sec:prelim}
\subsection{Notations}
Throughout the paper, we denote by $m$ and $n$ the number of constraints and the dimension of constraint matrices in SDPs, respectively. We use $\epsilon$ to denote the precision of the solution of the SDP feasibility problem in \cref{eqn:SDP-1} of \cref{defn:feasibility}. We use $r$ to denote an upper bound on the rank of matrices, i.e., $\max_{i\in\range{n}}\{\rank(A_{i}),\rank(C)\}\leq r$ (we denote $[n]:=\{1, \ldots, n\}$).

For a vector $v \in \bbc^n$, we use $\mathcal{D}_v$ to denote the probability distribution on $[n]$ where the probability of $i$ being chosen is $\mathcal{D}_v(i) = |v(i)|^2/\norm{v}^2$ for all $i \in [n]$. When it is clear from the context, a sample from $\mathcal{D}_v$ is often referred to as a sample from $v$. For a matrix $A \in \bbc^{n \times n}$, we use $\norm{A}$ and $\|A\|_F$ to denote its spectral norm and Frobenius norm, respectively; we use $A(i,\cdot)$ and $A(\cdot,j)$ to denote the $i^{\text{th}}$ row and $j^{\text{th}}$ column of $A$, respectively.  We use $\rows(A)$ to denote the $n$-dimensional vector formed by the norms of its row vectors, i.e., $\rows(A)(i) = \norm{A(i, \cdot)}$, for all $i \in [n]$.
\subsection{Data structure for sampling and query access}\label{sec:datastruct}
As we develop sublinear-time algorithms for solving SDP in this paper, the whole constraint matrices cannot be loaded into memory since storing them requires at least linear space and time. Instead, we assume the \emph{sampling access} of each constraint matrix as defined in \cref{defn:sampling-informal}. This sampling access relies on a natural probability distribution that arises in many machine learning applications~\cite{FKV04, chia2018quantum, GLT18, kerenidis2016recommendation, kerenidis2018quantum, tang2018quantum, tang2018quantum2} (also see a survey by Kannan and Vempala~\cite{kannan2017RandAlgNumLinAlg}).

Technically, Ref.~\cite{FKV04} used this sampling assumption to develop a sublinear algorithm for approximating low-rank projection of matrices. It is well-known (as pointed out by~\cite{kerenidis2016recommendation} and also used in~\cite{chia2018quantum, GLT18, kerenidis2018quantum, tang2018quantum, tang2018quantum2}) that there exist low-overhead preprocessed data structures that allow for the sampling access. More precisely, the existence of the data structures for the sampling access defined  in \cref{defn:sampling-informal} is stated as follows.

\begin{theorem}[\cite{kerenidis2016recommendation}]
  \label{thm:ds}
  Given a matrix $M \in \bbc^{n \times n}$ with $s$ non-zero entries, there exists a data structure storing $M$ in space $O(s \log n)$, which supports the following:
  \begin{enumerate}
    \item Read and write $M(i, j)$ in $O(\log n)$ time.
    \item Evaluate $\norm{M(i, \cdot)}$ in $O(\log n)$ time.
    \item Evaluate $\pnorm{F}{M}^2$ in $O(1)$ time.
    \item Sample a row index of $M$ according to statement 1 of \cref{defn:sampling-informal} in $O(\log n)$ time.
    \item For each row, sample an index according to statement 2 of \cref{defn:sampling-informal} in $O(\log n)$ time.
  \end{enumerate}
\end{theorem}

Readers may refer to~\cite[Theorem A.1]{kerenidis2016recommendation} for the proof of \cref{thm:ds}. In the following, we give the intuition of the data structure, which is demonstrated in \cref{fig:ds}. We show how to sample from a row vector: we use a binary tree to store the data of each row. The square of the absolute values of all entries, along with their original values are stored in the leaf nodes. Each internode contains the sum of the values of its two immediate children. It is easy to see that the root node contains the square of the norm of this row vector. To sample an index and to query an entry from this row, logarithmic steps suffice. To fulfill statement 1 of \cref{defn:sampling-informal}, we treat the norms of rows as a vector $(\norm{M(1, \cdot)}, \ldots, \norm{M(n, \cdot)})$ and organize the data of this vector in a binary tree.
\begin{figure}[ht]
  \centering
  \includegraphics[width=0.8\textwidth]{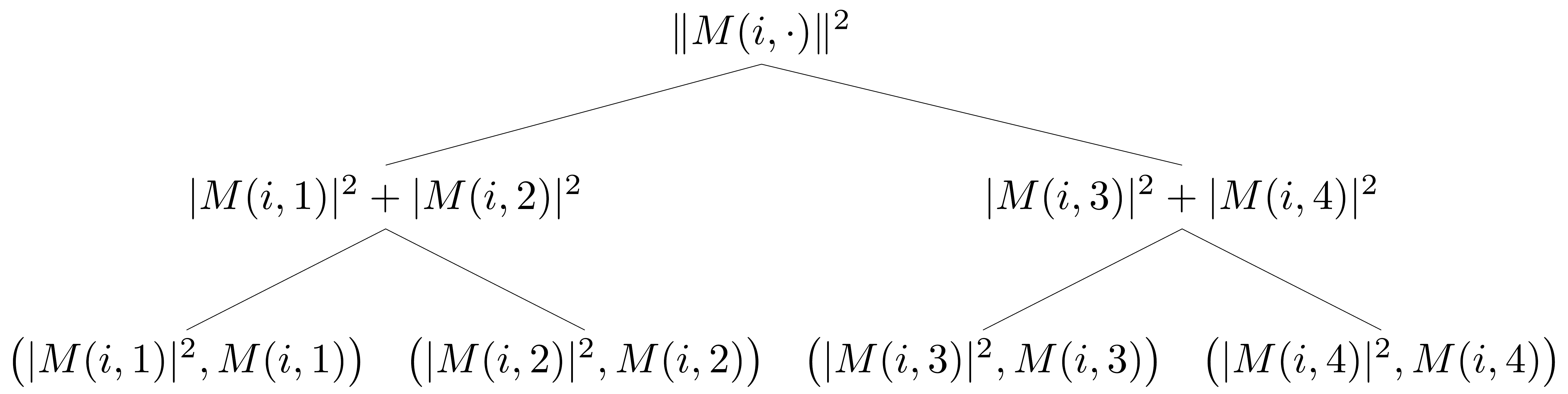}
  \caption{Illustration of a data structure that allows for sampling access to a row of $M \in \bbc^{4 \times 4}$. \label{fig:ds}}
\end{figure}

\subsection{Feasibility testing of SDPs}\label{sec:MMW-SDP}
We adopt the MMW framework to solve SDPs under the zero-sum approach \cite{Hazan,Wu10,gutoski2012parallel,LRS15,brandao2018SDP}. This is formulated as the following theorem:

\begin{theorem}[Master algorithm]\label{thm:SDP-master}
Feasibility of the SDP in \crefrange{eqn:SDP-1}{eqn:SDP-3} can be tested by \cref{algo:matrixMW-gain}.
\end{theorem}

\begin{algorithm}[htbp]
Set the initial Gibbs state $\rho_{1}=\frac{I_{n}}{n}$, and number of iterations $T=\frac{16\ln n}{\epsilon^{2}}$\;
\For{$t=1,\ldots,T$}{
Find a $j_{t}\in\range{m}$ such that $\Tr[A_{j_{t}}\rho_{t}]>a_{j_{t}}+\epsilon$. If we cannot find such $j_{t}$, claim that $\rho_{t}\in\mathcal{S}_{\epsilon}$ (the SDP is feasible) and terminate the algorithm\;
\label{line:observe-gain-matrix}
    Define the new weight matrix $W_{t+1}:=\exp[-\frac{\epsilon}{4}\sum_{i=1}^{t}A_{j_{i}}]$ and Gibbs state $\rho_{t+1}:=\frac{W_{t+1}}{\Tr[W_{t+1}]}$\label{line:Gibbs}\;
    }
Claim that the SDP is infeasible and terminate the algorithm\;
\caption{MMW for testing feasibility of SDPs.}
\label{algo:matrixMW-gain}
\end{algorithm}

This theorem is proved in \cite[Theorem 2.3]{brandao2018SDP}; note that the weight matrix therein is $W_{t+1}=\exp[\frac{\epsilon}{2}\sum_{i=1}^{t}M_{i}]$ where $M_{i}=\frac{1}{2}(I_{n}-A_{j_{i}})$, and this gives the same Gibbs state as in \cref{line:Gibbs} since for any Hermitian matrix $W\in\C^{n\times n}$ and real number $c\in\bbr$,
\begin{align}
\frac{e^{W+cI}}{\Tr[e^{W+cI}]}=\frac{e^{W}e^{c}I}{\Tr[e^{W}e^{c}I]}=\frac{e^{W}}{\Tr[e^{W}]}.
\end{align}
It should also be understood that this master algorithm is \emph{not} the final algorithm; the step of trace estimation with respect to the Gibbs state (in \cref{line:observe-gain-matrix} and \cref{line:Gibbs}) will be fulfilled by our sampling-based approach.


\section{The main algorithm}\label{sec:main_alg}
Our main algorithm is \cref{alg:fea_testing}, it depends on several building blocks. To begin with, we introduce a useful claim from~\cite[Lemma 11]{GLT18}.
\begin{claim}[Trace inner product estimation \cite{GLT18}]\label{claim:tr_prod}
Let $A\in \mathbb{C}^{n\times n}$ and $B\in \mathbb{C}^{n\times n}$ be two Hermitian matrices. Given sampling and query access to $A$ and query access to $B$. Then one can estimate $\Tr[AB]$ with the additive error $\epsilon_s$ with probability at least $1-\delta$ by using
\begin{align}
    O\Big(\frac{\|A\|_F^2\|B\|_F^2}{\epsilon_s^2} \big(Q(A)+Q(B)+S(A)+N(A)\big)\log\frac{1}{\delta}\Big)
\end{align}
time and queries, where $Q(B)$ is the cost of query access to $B$, and $Q(A), S(A), N(A)$ are the cost for query access, sampling access and row norm access to $A$.
\end{claim}

\subsection{Weighted sampling}
\label{sec:matrix_samp}

Given sample and query access to $A_1,\dots,A_{\tau}$, the objective of this section is to provide an algorithm to give sample and query access to a matrix $V$, which approximates the eigenvectors of $A := A_1+\cdots +A_{\tau}$. Specifically, we will show that $\|VV^{\dag} A VV^{\dag}-A\|_F$ can be bounded.

Note that trivially invoking the standard FKV sampling method~\cite{FKV04} is not capable of this task. In this paper, we propose the \emph{weighted sampling} method. The intuition is to assign each $A_{\ell}$ a different weight when computing the probability distribution, and then sampling a row/column index of $A$ according to this probability distribution. The main theorem we prove in this section is as follows.

We first give the method for sampling row indices of $A$ in \cref{proc:samp_row}. The objective of this procedure is to sample a submatrix $S$ such that $S^{\dag}S \approx A^{\dag}A$.

After applying \cref{proc:samp_row}, we obtain the row indices $i_1, \ldots, i_p$. Let $S_1, \ldots, S_{\tau}$ be matrices such that $S_{\ell}(t, \cdot) = A_{\ell}(i_t, \cdot)/\sqrt{p P_{i_t}}$ for all $t \in [p]$ and $\ell \in [\tau]$. Define $S$ as
\begin{align}
  \label{eq:s}
  S = S_1 + \cdots + S_{\tau}.
\end{align}
Next, we give the method for sampling column indices of $S$ as in \cref{proc:samp_col}: we need to sample a submatrix $W$ from $S$ such that $WW^{\dag} \approx SS^{\dag}$.

Now, we obtained column indices $j_1, \ldots, j_p$. Let $W_1, \ldots, W_{\tau}$ be matrices such that  $W_{\ell}(\cdot, t) = S_{\ell}(\cdot, j_t)/\sqrt{p P_{j_t}'}$ $\forall\,t \in [p], \ell \in [\tau]$, where $P_{j}' = \frac{1}{p}\sum_{t=1}^pQ_{j|i_t}$ for $j \in [n]$. Define $W$ as
\begin{align}
  \label{eq:w}
  W = W_1 + \cdots + W_{\tau}.
\end{align}

\begin{procedure}
  \SetKwInput{Input}{Input}\SetKwInOut{Output}{output}
  \Input{The sampling and query access to each $A_{\ell}$ as in \cref{defn:sampling-informal} for $A = \sum_{l=1}^{\tau}A_{\ell}$; integer $p$.}
  Sample $p$ indices $i_1, \ldots, i_p$ from $[n]$ according to the probability distribution $\{P_1, \ldots, P_n\}$ where $P_i = \sum_{j=1}^{\tau} \mathcal{D}_{\rows(A_j)}(i)\pnorm{F}{A_j}^2/\left(\sum_{\ell=1}^{\tau}\pnorm{F}{A_{\ell}}^2\right)$\label{step:it} by sampling $j$ according to $\pnorm{F}{A_j}^2/\left(\sum_{\ell=1}^{\tau}\pnorm{F}{A_{\ell}}^2\right)$ then sample according to $\mathcal{D}_{\rows(A_j)}$\;
  \caption{Weighted sampling of rows.()}
  \label[procedure]{proc:samp_row}
\end{procedure}

\begin{procedure}[ht]
  \SetKwInput{Input}{Input}\SetKwInOut{Output}{output}
  \Input{The sampling and query access to each $A_{\ell}$ as in \cref{defn:sampling-informal} for $A = \sum_{l=1}^{\tau}A_{\ell}$; $i_1, \ldots, i_p$ obtained in \cref{proc:samp_row}; integer $p$.}
  Do the following $p$ times independently to obtain samples $j_1, \ldots, j_p$. \Begin{
  Sample a row index $t \in [p]$ uniformly at random\;
  Sample a column index $j \in [n]$ from the probability distribution $\{Q_{1|i_t}, \ldots, Q_{n|i_t}\}$ where $Q_{j|i_t} = \sum_{k=1}^{\tau}\mathcal{D}_{A_{k}(i_t, \cdot)}(j)\norm{A_{k}(i_t, \cdot)}^2/\left(\sum_{\ell=1}^{\tau}\norm{A_{\ell}(i_t,\cdot)}^2\right)$ \label{step:q}\;
  }
  \caption{Weighted sampling of columns.()}
  \label[procedure]{proc:samp_col}
\end{procedure}

With the weighted sampling method, we obtained a small submatrix $W$ from $A$. Now, we use the singular values and singular vectors of $W$ to approximate the ones of $A$. This is shown in \cref{alg:approx_v}. The main consequence of \cref{alg:approx_v} is summarized in \cref{thm:vvavv_bound}.

\begin{algorithm}[ht]
  \SetKwInput{Input}{Input}\SetKwInOut{Output}{output}
  \Input{The sampling and query access to each $A_{\ell}$ as in \cref{defn:sampling-informal} for $A = \sum_{l=1}^{\tau}A_{\ell}$ with $\rank(A_{\ell}) \leq r$; error parameter $\epsilon$.}
  Set $p = 2\cdot10^{20}\frac{\tau^{12}r^{19}}{\epsilon^6}$, $\gamma = \frac{\epsilon^2}{3\times 10^6\tau^2r^6}$ \label{step:parameter}\;
  Use \cref{proc:samp_row} to obtain row indices $i_1, \ldots, i_p$\;
  Let $S_1, \ldots, S_{\tau}$ be matrices such that $S_{\ell}(t, \cdot) = A_{\ell}(i_t, \cdot)/\sqrt{p P_{i_t}}$ for all $t \in [p]$ and $\ell \in [\tau]$, where $P_i$ is defined in \cref{step:it} in \cref{proc:samp_row}. Let $S = S_1 + \cdots + S_{\tau}$ \label{step:alg-S}\;
  Use \cref{proc:samp_col} to obtain column indices $j_1, \ldots, j_p$\;
  Let $W_1, \ldots, W_{\tau}$ be matrices such that $W_{\ell}(\cdot, t) = S_{\ell}(\cdot, j_t)/\sqrt{p P_{j_t}'}$ for all $t \in [p]$ and $\ell \in [\tau]$, where $P_{j}' = \frac{1}{p}\sum_{t=1}^pQ_{j|i_t}$ for $j \in [n]$ and $Q_{j|i}$ is defined in \cref{step:q} in \cref{proc:samp_col}. Let $W = W_1 + \cdots + W_{\tau}$\;
  (Assume the rank of $A$ is $\hat{r}$.) Compute the top $\hat{r}$ singular values $\sigma_1, \ldots, \sigma_{\hat{r}}$ of $W$ and their corresponding left singular vectors $u_1, \ldots, u_{\hat{r}}$\;
  Discard the singular values and their corresponding singular vectors satisfying $\sigma_j^2 < \gamma\sum_{\ell=1}^\tau\pnorm{F}{W_{\ell}}^2$. Let the remaining number of singular values be $\tilde{r}$\;
  Output $i_1, \ldots, i_p$, $P_{i_1}, \ldots, P_{i_p}$, $\sigma_1, \ldots, \sigma_{\tilde{r}}$ and $u_1, \ldots, u_{\tilde{r}}$ \label{step:alg:v}\;
  \caption{Approximation of singular vectors.}
  \label{alg:approx_v}
\end{algorithm}

\begin{restatable}{theorem}{thmvvavv}
\label{thm:vvavv_bound}
Let $A = A_1 + \cdots + A_{\tau} \in \bbc^{n \times n}$ be a Hermitian matrix where $A_{\ell} \in \bbc^{n \times n}$ is Hermitian, $\norm{A_{\ell}} \leq 1$, and $\rank(A_{\ell}) \leq r$ for all $\ell \in [\tau]$. The sampling and query access to each $A_{\ell}$ is given as in \cref{defn:sampling-informal}. Take the error parameter $\epsilon$ as the input of \cref{alg:approx_v} to obtain the singular values $\sigma_1, \ldots, \sigma_{\tilde{r}}\in\bbr$ and singular vectors $u_1, \ldots, u_{\tilde{r}}\in\bbc^{p}$ for $p$ specified in \cref{step:parameter} of \cref{alg:approx_v}. Let $V \in \bbc^{n \times \tilde{r}}$ be the matrix such that $V(\cdot, j) = \frac{S^{\dag}}{\sigma_j}u_j$ for $j \in \{1, \ldots, \tilde{r}\}$, where $S$ is defined in \cref{step:alg-S} in \cref{alg:approx_v}. Then with probability at least $9/10$, it holds that $\pnorm{F}{VV^{\dag}AVV^{\dag} - A} \leq \ef(2+\ev)$.
\end{restatable}

We leave the proofs of \cref{thm:vvavv_bound} and the weighted sampling algorithm to \cref{sec:matrix_samp_proof}.


\subsection{Symmetric approximation of low-rank Hermitian matrices}\label{sec:approx-eigen}
Now we show that the spectral decomposition of the sum of low-rank Hermitian matrices can be approximated in time logarithmic in the dimension with the given data structure. We call this technique \emph{symmetric approximation}.

Briefly speaking, suppose we are given the approximated left singular vectors $V$ of $A$ outputted by \cref{alg:approx_v} such that $\norm{VV^\dag AVV^\dag-A}_{F}$ is bounded as in \cref{thm:vvavv_bound}, then we can approximately do spectral decomposition of $A$ as follows. First, we approximate the matrix $V^\dag A V$ by sampling. Then, since $V^\dag AV$ is a matrix with low dimension, we can do spectral decomposition of the matrix efficiently as $UDU^\dag$ (see \cref{lem:appx_vav}). Finally, we show that $VU$ is close to an isometry. Therefore, $(VU)D(VU)^\dag$ is an approximation to the spectral decomposition of $A$.

\begin{restatable}{lemma}{vavb}
    \label{lem:appx_vav}
Let $V\in \mathbb{C}^{n\times \tilde{r}}$ and $A = \sum_{\ell}^{\tau} A_{\ell} \in \mathbb{C}^{n\times n}$ be a Hermitian matrix. Given query access and sampling access to $A_{\ell}$ for $\ell\in[\tau]$, where each query and sample take $O(\log n)$ time , and query access to $V$ , where each query takes $O(p)$ time. Let $r$ be some integer such that $r\geq \tilde{r}$, and let $\epsilon_s = \frac{\epsilon}{400r^2}$.
Then, one can output a Hermitian matrix $\tilde{B}\in \mathbb{C}^{\tilde{r}\times \tilde{r}}$ such that $\|V^{\dag}AV-\tilde{B}\|_F \leq \epsilon_s$ with probability $1-\delta$ by using $O\L((p+\log n) \frac{r^{5}\tau^3}{\epsilon^2}\log\frac{1}{\delta}\R)$ time.
\end{restatable}

The algorithm for approximating the spectral decomposition of $A$ is described in \cref{alg:appx_VVAVV}, and the effectiveness of \cref{alg:appx_VVAVV} is summarized in \cref{thm:our_appx}:
\begin{algorithm}
  \SetKwInput{Input}{Input}\SetKwInOut{Output}{output}
  \Input{The sampling and query access to each $A_{\ell}$ as in \cref{defn:sampling-informal} for $A = \sum_{l=1}^{\tau}A_{\ell}$ and query access to $V\in \bbc^{n \times \tilde{r}}$ (obtained from \cref{alg:approx_v}, also see \cref{thm:vvavv_bound}); error parameter $\epsilon$.}
  Compute a matrix $\tilde{B}\in \bbc^{\tilde{r} \times \tilde{r}}$ whose spectrum approximates that of $A$ (this is achieved by \cref{lem:appx_vav})\;
  Compute the spectral decomposition $UDU^\dag$ of matrix $\tilde{B}$\;
  Output an isometry $U\in \bbc^{\tilde{r} \times \tilde{r}}$ and a diagonal matrix $D\in \bbc^{\tilde{r} \times \tilde{r}}$ such that $UDU^\dag$ is the spectral decomposition of $\tilde{B}$.
  \caption{Approximation of the spectral decomposition of $A$.}
  \label{alg:appx_VVAVV}
\end{algorithm}

\begin{restatable}{lemma}{thmsvd}
\label{thm:our_appx}
\cref{alg:appx_VVAVV} outputs a Hermitian matrix $\tilde{B}\in \mathbb{C}^{\tilde{r}\times \tilde{r}}$ with probability at least $1-\delta$ with time and query complexity $O((p+\log n) \frac{r^{5}\tau^3}{\epsilon^2}\log\frac{1}{\delta})$ such that
  \begin{align}
    \|V\tilde{B}V^\dag - A\|\leq \Bigl(1+\ev\Bigr)^2\es + \Bigl(2+\ev\Bigr)\ef.
  \end{align}
\end{restatable}

We leave the proof of \cref{thm:our_appx} to the end of  \cref{sec:matrix_samp_proof}.


\subsection{Approximating Gibbs states}
\label{sec:gibbs}
In this subsection, we combine our techniques from \cref{sec:matrix_samp} and \cref{sec:approx-eigen} to give a sampling-based estimator of the traces of a Gibbs state times a constraint $A_\ell$. This is formulated as \cref{alg:appx_tr}.

\begin{algorithm}
  \SetKwInput{Input}{Input}\SetKwInOut{Output}{output}
  \Input{The sampling and query access to each $A_{\ell}$ as in \cref{defn:sampling-informal} for $A = \sum_{l=1}^{\tau}A_{\ell}$; query access to $V\in \bbc^{n \times \tilde{r}}$ (obtained from \cref{alg:approx_v}, also see \cref{thm:vvavv_bound}); matrices $U$ and $D$ (obtained from \cref{alg:appx_VVAVV}) such that $(VU)D(VU)^\dag$ is an approximated spectral decomposition of $A$ as in \cref{thm:our_appx}.}
  Compute $\eta=\Tr[e^{-\frac{\epsilon}{2} D}]$\;
  Approximate $\Tr[A_\ell (VU)(e^{-\frac{\epsilon}{2} D}/\Tr[e^{-\frac{\epsilon}{2} D}])(VU)^\dag]$ by $\zeta$ according to \cref{claim:tr_prod}\;
  Output $\zeta$, $\eta$.
  \caption{Approximation of the trace.}
  \label{alg:appx_tr}
\end{algorithm}

We show that the output of \cref{alg:appx_tr} $\epsilon$-approximates $\Tr[A_\ell\rho]$ for $\rho=e^{-\frac{\epsilon}{2}A}/\Tr[e^{-\frac{\epsilon}{2}A}]$. Let $\tilde{A} = VV^\dag AVV^\dag$. Let $U$ and $D$ be outputs of \cref{alg:appx_VVAVV}, which will be used in \cref{alg:appx_tr}. We suppose $\norm{\tilde{A}-A}_{F}\leq (2+\ev)\ef$ as in \cref{thm:vvavv_bound}.

\begin{restatable}{lemma}{estgibbs}
\label{thm:est_gibbs}
Let $\rho = \frac{e^{-\frac{\epsilon}{2} A}}{\Tr{e^{-\frac{\epsilon}{2} A}}}$ and $\hat{\rho} = \frac{(VU)e^{-\frac{\epsilon}{2} D}(VU)^\dag}{\Tr{e^{-\frac{\epsilon}{2} D}}}$. Suppose $\|A - \tilde{A}\|_F\leq (2+\ev)\ef$. Let $A_\ell$ be a Hermitian matrix with the promise that $\|A_\ell\| \leq 1$ and $\rank(A_\ell)\leq r$. Then \cref{alg:appx_tr} outputs $\zeta$ such that
\begin{align}
    |\Tr{[A_\ell\rho]}-\zeta| \leq \epsilon
\end{align}
with probability $1-\delta$ in time $O(\frac{4}{\epsilon^2}(\log n+\tau pr)\log\frac{1}{\delta})$.
\end{restatable}
We leave the proof of \cref{thm:est_gibbs} to \cref{sec:gibbs_proof}.

\subsection{Proof of the main algorithm}\label{sec:proof-main}
We finally state our main result on solving SDPs via sampling.
\begin{theorem}~\label{thm:main}
Given Hermitian matrices $\{A_1,\dots,A_m\}$ with the promise that each of them has rank at most $r$, spectral norm at most $1$, and the sampling access of each $A_i$ is given by \cref{defn:sampling-informal}. Also given $a_{1},\ldots,a_{m}\in\bbr$ and $\epsilon>0$. Then \cref{alg:fea_testing} gives a succinct description and any entry (see \cref{remark:SDP-solution}) of the solution of the SDP feasibility problem
\begin{align}
  \label{eqn:SDP-formal}
\Tr[A_{i} X]\leq a_{i}+\epsilon\ \  \forall\,i\in\range{m};\quad X\succeq 0;\quad
\Tr[X]=1
\end{align}
with probability at least $2/3$ in $O(\frac{mr^{57}\ln^{37}n}{\epsilon^{92}})$ time.
\end{theorem}

\begin{algorithm}
Set the initial Gibbs state $\rho_{1}=\frac{I_{n}}{n}$, and number of iterations $T=\frac{16\ln n}{\epsilon^{2}}$\;
\For{$t=1,\ldots,T$}{
Find a $j_{t}\in\range{m}$ such that $\Tr[A_{j_{t}}\rho_{t}]>a_{j_{t}}+\epsilon$ using \cref{alg:appx_tr}. If we cannot find such $j_{t}$, claim that $\rho_{t}\in\mathcal{S}_{\epsilon}$ and terminate the algorithm. Output $\rho_{t}(\ell, j) = \sum_{k=1}^{\tilde{r}}V(\ell, k)e^{\sigma_k\epsilon/2}V(j, k)^*/\eta$, where $V(\ell, j) = \sum_{s=1}^p \frac{A^*(i_s, \ell)u_j(t)}{\sqrt{P_{i_s}}\sigma_j}$, $i_1, \ldots, i_p$, $P_{1_1}, \ldots, P_{i_p}$, $\sigma_1, \ldots, \sigma_{\tilde{r}}$ and $u_1, \ldots, u_{\tilde{r}}$ are obtained from \cref{alg:approx_v} and $\eta$ is obtained from \cref{alg:appx_tr}\;
    Define the new weight matrix $W_{t+1}:=\exp[-\frac{\epsilon}{4}\sum_{i=1}^{t}A_{j_{i}}]$ and Gibbs state $\rho_{t+1}:=\frac{W_{t+1}}{\Tr[W_{t+1}]}$\;
    }
Claim that the SDP is infeasible and terminate the algorithm\;
\caption{Feasibility testing of SDPs by our sampling-based approach.}
\label{alg:fea_testing}
\end{algorithm}
The algorithm follows the master algorithm in \cref{thm:SDP-master}. The main challenge is to estimate $\Tr[A_{j_t}\rho_t]$ where $\rho_t$ is the Gibbs state at iteration $t$; this is achieved by \cref{thm:est_gibbs} in \cref{sec:gibbs}.

\begin{proof}
\textbf{Correctness:} The correctness of \cref{alg:fea_testing} directly follows from \cref{thm:est_gibbs}. Specifically, we have shown that one can estimate the quantity $\Tr[A_{j_t}\rho_t]$ with precision $\epsilon$ with high probability by applying \cref{alg:approx_v,alg:appx_VVAVV,alg:appx_tr}.

\hd{Time complexity:} First, we show that given the data structure in \cref{thm:ds}, \cref{alg:approx_v} can be computed in time $O(p^3+p\tau\log \tau \log n)$. \cref{proc:samp_row} and \cref{proc:samp_col} both can be done in time $O(p\tau\log \tau \log n)$. There are many ways to implement \cref{proc:samp_row}. For example, build a binary tree as in \cref{thm:ds} for $\norm{A_1}_F, \ldots, \norm{A_\tau}_F$ to sample $j\in[\tau]$ according to $\norm{A_j}_F^2/(\sum_{\ell=1}^\tau\norm{A_\ell}_F^2)$, and then use the data structures in \cref{thm:ds} to sample from $\mathcal{D}_{\rows(A_j)}$. The time complexity is $O(\tau\log \tau \log n)$. Similarly, we can implement \cref{proc:samp_col} in time $O(\tau\log \tau \log n)$. Hence, the time complexity to construct the matrix $W$ and compute its SVD is $O(p\tau\log \tau \log n+p^3)$. \cref{alg:approx_v} succeeds with probability $9/10$.

Then, by \cref{thm:our_appx} and \cref{thm:est_gibbs},   \cref{alg:appx_VVAVV,alg:appx_tr} takes time $O((p+\log n) \frac{r^{5}\tau^3}{\epsilon^2}\log\frac{1}{\delta})$ and $O(\frac{4}{\epsilon^2}(\log n+\tau pr)\log\frac{1}{\delta})$ respectively to succeed with probability at least $1-\delta$.  Recall from the previous paragraph,  \cref{alg:approx_v} takes time $O(p\tau\log \tau \log n+p^3)$. In total, \cref{alg:appx_VVAVV,alg:appx_tr,alg:approx_v}  are each called $Tm=\frac{16 m \ln n}{\eps^2}$ times. We specify that
\begin{itemize}
    \item $\tau=T=\frac{16\ln n}{\epsilon^{2}}$, and
    \item $p = 2\cdot10^{20}\frac{\tau^{12}r^{19}}{\epsilon^6}$
\end{itemize}
(see also \cref{alg:approx_v}). By setting $\delta$ as a small enough constant (say $\delta=1/6$) and noticing the $p$ dominates other terms,  \cref{alg:fea_testing} succeeds with probability at least $2/3$ in time $O(T m p^3) = O\L(\frac{mr^{57}\ln^{37}n}{\epsilon^{92}}\R)$.
\end{proof}

\begin{remark}\label{remark:SDP-solution}
\cref{thm:main} solves the SDP feasibility problem, i.e., to decide $\mathcal{S}_{0}=\emptyset$ or $\mathcal{S}_{\epsilon}\neq\emptyset$. For the SDP optimization problem in \crefrange{eqn:SDP-OPT}{eqn:SDP-X}, an approximation to the optimal value can be found by a binary search with feasibility subroutines. (see \cref{footnote:binary-search}); however, writing down the approximate solution would take $n^{2}$ space, ruining the poly-logarithmic complexity in $n$. Nevertheless,
\begin{itemize}
    \item we have its \emph{succinct representation} $i_1, \ldots, i_p$, $P_{i_1}, \ldots, P_{i_p}$, $\sigma_1, \ldots, \sigma_{\tilde{r}}$, $u_1, \ldots, u_{\tilde{r}}$, and $\eta$;
    \item we can compute any entry of the solution matrix according to this succinct description as in Step 3 of \cref{alg:fea_testing}.
\end{itemize}
With the succinct description given in the first statement, one can perform other operations on the solution matrix using the similar sampling-based methods.
\end{remark}

\begin{remark}\label{remark:practicality}
The large polynomial overhead in \cref{thm:main} may not be necessary in practice and can potentially be reduced by more fine-grained analysis. This is also suggested by numerics in practice (see~\cite{arrazola2019quantum}).
\end{remark}

\hd{Application to shadow tomography.}
As a direct corollary of \cref{thm:main}, we have:
\begin{restatable}{corollary}{shadowtomography}
  \label{cor:shadow}
Given Hermitian matrices $\{E_1,\dots,E_m\}$ with the promise that each of $E_1,\dots,E_m$ has rank at most $r$, $0\preceq E_i\preceq I$ and the sampling access to $E_i$ is given as in \cref{defn:sampling-informal} for all $i\in\range{m}$. Also given $p_{1},\ldots,p_{m}\in\bbr$. Then for any $\epsilon>0$, the shadow tomography problem
\begin{align}
  \label{eqn:shadow-formal}
  \text{Find $\sigma$ such that}\quad |\Tr[\sigma E_{i}]-p_{i}|\leq\epsilon\ \ \forall\,i\in\range{m}\ \ \text{subject to } \sigma\succeq 0,\ \ \Tr[\sigma]=1
\end{align}
can be solved with probability $1-\delta$ with cost $O(m\cdot\poly(\log n, 1/\epsilon,\log(1/\delta),r))$.
\end{restatable}

The proof of this corollary is deferred to \cref{sec:proof-shadow}. Here, $p_{i}=\Tr[\rho E_{i}]$ in \cref{eqn:SDP-shadow-1} for all $i\in\range{m}$. Notice that the assumption of knowing $p_{1},\ldots,p_{m}$ makes our problem slightly different from the shadow tomography problem in \cite{aaronson2017quantum,brandao2018SDP,vanApeldoorn2018SDP} where we are only given copies of the quantum state $\rho$ without the knowledge of $\Tr[\rho E_{1}],\ldots,\Tr[\rho E_{m}]$. However, quantum state is a concept without a counterpart in classical computing, hence we follow the conventional assumption in SDPs that these real numbers are given.

\begin{remark}\label{remark:shadow}
Similar to \cref{remark:SDP-solution}, $\sigma$ can be stored as a succinct representation because we can  have $\sigma=\frac{\exp[\frac{\epsilon}{2}\sum_{\tau=1}^{t}(-1)^{i_{\tau}}A_{j_{\tau}}]}{\Tr\big[\exp[\frac{\epsilon}{2}\sum_{\tau=1}^{t}(-1)^{i_{\tau}}A_{j_{\tau}}]\big]}$ in \cref{cor:shadow}, where $t\leq T$ and $i_{\tau}\in\{0,1\}$, $j_{\tau}\in\range{m}$ for all $\tau\in\range{t}$. Storing all $i_{\tau},j_{\tau}$ takes $t(\log_{2}m+1)=O(\log m\log n/\epsilon^{2})$ bits.
\end{remark}

\section*{Acknowledgement}
We thank Scott Aaronson, Andr{\'a}s Gily{\'e}n, Ewin Tang, Ronald de Wolf, and anonymous reviewers for their detailed feedback on preliminary versions of this paper. NHC, HHL, and CW were supported by Scott Aaronson's Vannevar Bush Faculty Fellowship. TL was supported by IBM PhD Fellowship, QISE-NET Triplet Award (NSF DMR-1747426), and the U.S. Department of Energy, Office of Science, Office of Advanced Scientific Computing Research, Quantum Algorithms Teams program.



\appendix
\section{Proof of \cref{thm:vvavv_bound}}\label{sec:matrix_samp_proof}
Before showing $S^{\dag}S \approx A^{\dag}A$ and $SS^{\dag} \approx WW^{\dag}$, we first prove the following general result.

\begin{lemma}
  \label{lemma:mmnn}
  Let $M_1,\ldots, M_{\tau} \in \bbc^{n \times n}$ be a matrices. Independently sample $p$ rows indices $i_1, \ldots, i_p$ from $M = M_1 + \cdots + M_{\tau}$ according to the probability distribution $\{P_1, \ldots, P_n\}$ where
  \begin{align}
    \label{eq:prob-cond}
    P_i \geq \frac{\sum_{j=1}^{\tau} \mathcal{D}_{\rows(M_j)}(i)\pnorm{F}{M_j}^2}{(\tau+1)\sum_{\ell=1}^{\tau}\pnorm{F}{M_{\ell}}^2}.
  \end{align}
  Let $N_1, \ldots, N_{\tau} \in \bbc^{p \times n}$ be matrices with
  \begin{align}
    N_{\ell}(t, \cdot) = \frac{M_{\ell}(i_t,\cdot)}{\sqrt{pP_{i_t}}},
  \end{align}
  for $t \in [p]$ and $\ell \in [\tau]$. Define $N = N_1 + \cdots +N_{\tau}$. Then for all $\theta > 0$, it holds that
  \begin{align}
    \Pr\left(\pnorm{F}{M^{\dag}M - N^{\dag}N} \geq \theta\sum_{\ell=1}^{\tau}\pnorm{F}{M_{\ell}}^2\right) \leq \frac{(\tau+1)^2}{\theta^2p}.
  \end{align}
\end{lemma}

\begin{proof}
  We first show that the expected value of each entry of $N^{\dag}N$ is the corresponding entry of $M^{\dag}M$ as follows.
  \begin{align}
    \E\left(N^{\dag}(i, \cdot)N(\cdot, j)\right) &= \sum_{t=1}^p\E\left(N^*(t, i)N(t, j)\right) \\
                                                 &= \sum_{t=1}^p\sum_{k=1}^n P_k\frac{M^*(k, i)M(k, j)}{pP_k} \\
                                                 &= M^{\dag}(i, \cdot)M(\cdot, j).
  \end{align}
  \begin{align}
    &\E\left(|N^{\dag}(i, \cdot)N(\cdot, j) - M^{\dag}(i, \cdot)M(\cdot, j)|^2\right) \leq \sum_{t=1}^p \E \left((N^*(t, i) N(t, j))^2\right) \\
                                                                                     &= \sum_{t=1}^p\sum_{k=1}^nP_k\frac{(M^*(k, i))^2(M(k, j))^2}{p^2P_k^2} \\
                                                                                     &\leq \frac{(\tau+1)\sum_{\ell=1}^{\tau}\pnorm{F}{M_{\ell}}^2}{p}\sum_{k=1}^n\frac{(M^*(k, i))^2(M(k, j))^2}{\sum_{\ell'=1}^{\tau}\mathcal{D}_{\rows(M_{\ell'})}(k)\pnorm{F}{M_{\ell'}}^2} \\
                                                                                     &= \frac{(\tau+1)\sum_{\ell=1}^{\tau}\pnorm{F}{M_{\ell}}^2}{p}\sum_{k=1}^n\frac{(M^*(k, i))^2(M(k, j))^2}{\sum_{\ell'=1}^{\tau}\norm{M_{\ell'}(k, \cdot)}^2}
  \end{align}
  Now, we bound the expected distance between $N^{\dag}N$ and $M^{\dag}M$:
  \begin{align}
    \E\left(\pnorm{F}{M^{\dag}M - N^{\dag}N}^2\right) &= \sum_{i,j=1}^n\E\left(|N^{\dag}(i, \cdot)N(\cdot, j) - M^{\dag}(i, \cdot)M(\cdot, j)|^2\right) \\
                                                      &\leq \frac{(\tau+1)\sum_{\ell=1}^{\tau}\pnorm{F}{M_{\ell}}^2}{p}\sum_{k=1}^n\frac{\sum_{i,j=1}^n(M^*(k, i))^2(M(k, j))^2}{\sum_{\ell'=1}^{\tau}\norm{M_{\ell'}(k, \cdot)}^2} \\
                                                      &= \frac{(\tau+1)\sum_{\ell=1}^{\tau}\pnorm{F}{M_{\ell}}^2}{p}\sum_{k=1}^n\frac{\norm{M(k, \cdot)}^4}{\sum_{\ell'=1}^{\tau}\norm{M_{\ell'}(k, \cdot)}^2} \\
                                                      &= \frac{(\tau+1)\sum_{\ell=1}^{\tau}\pnorm{F}{M_{\ell}}^2}{p}\sum_{k=1}^n\frac{\norm{M(k, \cdot)}^2\left(\tau\sum_{\ell'=1}^{\tau}\norm{M_{\ell'}(k, \cdot)}^2\right)}{\sum_{\ell''=1}^{\tau}\norm{M_{\ell''}(k, \cdot)}^2} \\
&= \frac{\tau(\tau+1)\sum_{\ell=1}^{\tau}\pnorm{F}{M_{\ell}}^2}{p}\pnorm{F}{M}^2 \\
&\leq \frac{\tau(\tau+1)\left(\sum_{\ell=1}^{\tau}\pnorm{F}{M_{\ell}}^2\right)^2}{p} \\
&\leq \frac{(\tau+1)^2\left(\sum_{\ell=1}^{\tau}\pnorm{F}{M_{\ell}}^2\right)^2}{p}.
  \end{align}
Consequently, the result of this lemma follows from Markov's inequality.
\end{proof}

The following technical claim will be used multiple times in this paper. It relates the three quantities: $\sum_{\ell=1}^{\tau}\pnorm{F}{A_{\ell}}^2$, $\sum_{\ell=1}^{\tau}\pnorm{F}{S_{\ell}}^2$, and $\sum_{\ell=1}^{\tau}\pnorm{F}{W_{\ell}}^2$:
\begin{claim}
  \label{claim:asa}
  Let $A = A_1 + \cdots + A_m$ be a matrix with the sampling access for each $A_{\ell}$ as in \cref{defn:sampling-informal}. Let $S$ and $W$ be defined by \cref{eq:s,eq:w}. Then, with probability at least $1-2\tau^2/p$ it holds that
  \begin{align}
    \label{eq:asa}
    \frac{1}{\tau+1} \sum_{\ell=1}^{\tau}\pnorm{F}{A_{\ell}}^2 \leq \sum_{\ell=1}^{\tau}\pnorm{F}{S_{\ell}}^2 \leq \frac{2\tau+1}{\tau+1} \sum_{\ell=1}^{\tau}\pnorm{F}{A_{\ell}}^2,
  \end{align}
  and
  \begin{align}
    \label{eq:sws}
    \frac{1}{\tau+1} \sum_{\ell=1}^{\tau}\pnorm{F}{S_{\ell}}^2 \leq \sum_{\ell=1}^{\tau}\pnorm{F}{W_{\ell}}^2 \leq \frac{2\tau+1}{\tau+1} \sum_{\ell=1}^{\tau}\pnorm{F}{S_{\ell}}^2,
  \end{align}
\end{claim}
\begin{proof}
  We first evaluate $\E(\normalpnorm{F}{S_{\ell}}^2)$ as follows. For all $\ell \in [\tau]$,
  \begin{align}
    \E\left(\pnorm{F}{S_{\ell}}^2\right) = \sum_{i=1}^p\E\left(\norm{S_{\ell}(i, \cdot)}^2\right) = \sum_{i=1}^p \sum_{j=1}^n P_{j}\frac{\norm{A_{\ell}(j, \cdot)}^2}{pP_{j}} = \sum_{j=1}^n\norm{A_{\ell}(j, \cdot)}^2 = \pnorm{F}{A_{\ell}}^2.
  \end{align}
  Then we have
  \begin{align}
    \norm{S_{\ell}(i, \cdot)}^2 = \sum_{j=1}^n\frac{|A_{\ell}(i, j)|^2}{pP_i} &\leq \sum_{j=1}^n\frac{2|A_{\ell}(i, j)|^2\sum_{\ell=1}^{\tau}\pnorm{F}{A_{\ell}}^2}{p\sum_{j=1}^{\tau}\norm{A_j(i, \cdot)}^2} \\
                                                                              &= \frac{2\norm{A_{\ell}(i, \cdot)}^2\sum_{\ell=1}^{\tau}\pnorm{F}{A_{\ell}}^2}{p\sum_{j=1}^{\tau}\norm{A_j(i, \cdot)}^2} \leq \frac{2\sum_{\ell=1}^{\tau}\pnorm{F}{A_{\ell}}^2}{p}.
  \end{align}
  Note that the quantity $\pnorm{F}{S_{\ell}}^2$ can be viewed as a sum of $p$ independent random variables $\norm{S_{\ell}(1, \cdot)}^2, \ldots, \norm{S_{\ell}(p, \cdot)}^2$. As a result,
  \begin{align}
    \Var(\pnorm{F}{S_{\ell}}^2) = p\Var(\norm{S_{\ell}(i, \cdot)}^2) &\leq p\E\left(\norm{S_{\ell}(i, \cdot)}^4\right) \\
                                                                     & \leq p\sum_{i=1}^nP_i\left(\frac{2\sum_{j=1}^{\tau}\pnorm{F}{A_j}^2}{p}\right)^2 = \frac{2\left(\sum_{\ell=1}^{\tau}\pnorm{F}{A_{\ell}}^2\right)^2}{p}.
  \end{align}
  According to Chebyshev's inequality, we have
  \begin{align}
    \Pr\left(\left|\pnorm{F}{S_{\ell}}^2 - \pnorm{F}{A_{\ell}}^2\right| \geq \frac{\sum_{\ell=1}^{\tau}\pnorm{F}{A_{\ell}}^2}{\tau}\right) \leq \frac{2\left(\sum_{\ell=1}^{\tau}\pnorm{F}{A_{\ell}}^2\right)^2}{p} = \frac{2\tau^2}{p}.
  \end{align}
  Therefore, with probability at least $1-\frac{2\tau^2}{p}$, it holds that
  \begin{align}
    -\frac{1}{\tau+1}\sum_{j\neq \ell}\pnorm{F}{A_j}^2 + \frac{\tau}{\tau+1}\pnorm{F}{A_{\ell}}^2
\leq \pnorm{F}{S_{\ell}}^2 \leq \frac{1}{\tau+1}\sum_{j\neq \ell}\pnorm{F}{A_j}^2 + \frac{\tau+2}{\tau+1}\pnorm{F}{A_{\ell}}^2,
  \end{align}
  which implies that
  \begin{align}
    \frac{1}{\tau+1} \sum_{\ell=1}^{\tau}\pnorm{F}{A_{\ell}}^2 \leq \sum_{\ell=1}^n\pnorm{F}{S_{\ell}}^2 \leq \frac{2\tau+1}{\tau+1} \sum_{\ell=1}^{\tau}\pnorm{F}{A_{\ell}}^2.
  \end{align}
  \cref{eq:sws} can be proven in a similar way.
\end{proof}

Now, the main result of the weighted sampling method, namely $A^{\dag}A \approx S^{\dag}S$ and $WW^{\dag} \approx SS^{\dag}$, is a consequence of \cref{lemma:mmnn}:

\begin{corollary}
  \label{cor:aass}
  Let $A = A_1 + \cdots + A_m$ be a matrix with the sampling access for each $A_{\ell}$ as in \cref{defn:sampling-informal}. Let $S$ and $W$ be defined by \cref{eq:s,eq:w}. Letting $\theta = (\tau+1)\sqrt{\frac{100}{p}}$, then, with probability at least $9/10$, the following holds:
  \begin{align}
    \label{eq:aass}
    \pnorm{F}{A^{\dag}A - S^{\dag}S} &\leq \theta \sum_{\ell = 1}^{\tau}\pnorm{F}{A_{\ell}}^2, \text{ and}\\
    \label{eq:ssww}
    \pnorm{F}{SS^{\dag} - WW^{\dag}} &\leq \theta \sum_{\ell=1}^{\tau}\pnorm{F}{S_{\ell}}^2 \leq 2\theta\sum_{\ell=1}^{\tau}\pnorm{F}{A_{\ell}}^2.
  \end{align}
\end{corollary}
\begin{proof}
  First note that \cref{eq:aass} follows from \cref{lemma:mmnn}. For the second statement, we need the probability $P'_j$ to satisfy \cref{eq:prob-cond}\ in \cref{lemma:mmnn}. In fact,
  \begin{align}
    P_j' = \sum_{t=1}^p\frac{Q_{j|i_t}}{p} &= \frac{1}{p}\sum_{t=1}^p\frac{\sum_{k=1}^{\tau}\mathcal{D}_{A_k(i_t, \cdot)}(j)\norm{A_k(i_t, \cdot)}^2}{\sum_{\ell=1}^{\tau}\norm{A_{\ell}(i_t, \cdot)}^2} \\
                                           &= \frac{1}{p}\sum_{t=1}^p\frac{\sum_{k=1}^{\tau}|A_k(i_t, j)|^2}{\sum_{\ell=1}^{\tau}\norm{A_{\ell}(i_t, \cdot)}^2} \\
                                           & = \frac{1}{p}\sum_{t=1}^p\frac{pP_{i_t}\sum_{k=1}^{\tau}|S_k(i_t, j)|^2}{\sum_{\ell=1}^{\tau}\norm{A_{\ell}(i_t, \cdot)}^2} \\
                                           & = \sum_{t=1}^p\frac{\sum_{j=1}^{\tau}\norm{A_j(i_t, \cdot)}^2}{\sum_{\ell=1}^{\tau}\pnorm{F}{A_{\ell}}^2}\frac{\sum_{k=1}^{\tau}|S_k(i_t, j)|^2}{\sum_{\ell=1}^{\tau}\norm{A_{\ell}(i_t, \cdot)}^2} \\
                                           & = \sum_{t=1}^p\frac{\sum_{k=1}^{\tau}|S_k(i_t, j)|^2}{\sum_{\ell=1}^{\tau}\pnorm{F}{A_{\ell}}^2} \\
                                           & = \frac{\sum_{k=1}^{\tau}\norm{S_k(\cdot, j)}^2}{\sum_{\ell=1}^{\tau}\pnorm{F}{A_{\ell}}^2} \\
                                           & \geq \frac{\sum_{k=1}^{\tau}\norm{S_k(\cdot, j)}^2}{(\tau+1)\sum_{\ell=1}^{\tau}\pnorm{F}{S_{\ell}}^2},
    \end{align}
    where the last inequality follows from \cref{claim:asa}. Note that the probability satisfies \cref{eq:prob-cond}; as a result of \cref{lemma:mmnn}, \cref{eq:ssww} holds.
\end{proof}

An important result of \cref{alg:approx_v} is that the vectors $u_1, \ldots, u_{\tilde{r}}$ are approximately orthonormal, as stated in the following lemma:
\begin{lemma}
  \label{lemma:almost-orthonormal}
  Let $A = A_1 + \cdots + A_{\tau}$ be a matrix with the sampling access to each $A_{\ell}$ as in \cref{defn:sampling-informal}. Assume $\norm{A_{\ell}} \leq 1$ and $\rank(A_{\ell}) \leq r$ for all $\ell \in [\tau]$. Take $A$ and error parameter $\epsilon$ as the input of \cref{alg:approx_v} and obtain the $\sigma_1, \ldots, \sigma_{\tilde{r}}$ and $u_1, \ldots, u_{\tilde{r}}$. Let $V \in \bbc^{n \times \tilde{r}}$ be the matrix such that $V(\cdot, j) = \frac{S^{\dag}}{\sigma_j}u_j$ for $j \in \{1, \ldots, \tilde{r}\}$. Then, with probability at least $9/10$, the following statements hold:
  \begin{enumerate}
    \item There exists an isometry $U \in \bbc^{n \times {\tilde{r}}}$ whose column vectors span the column space of $V$ satisfying $\pnorm{F}{U - V} \leq \ev$.
    \item $|\norm{V} - 1| \leq \ev$.
    \item Let $\Pi_V$ be the projector on the column space of $V$, then it holds that $\pnorm{F}{VV^{\dag} - \Pi_V} \leq \ev$.
    \item $\pnorm{F}{V^{\dag}V - I} \leq \ev$.
  \end{enumerate}
\end{lemma}

\begin{proof}
      Most of the arguments in this proof are similar to the proofs in \cite[Lemma 6.6, Corollary 6.7, Proposition 6.11]{tang2018quantumv1}. Let $v_j \in \bbc^n$ denote the column vector $V(\cdot, j)$, i.e., $v_j = \frac{S^{\dag}}{\hat{\sigma}_j}\hat{u}_j$. Choose $\theta = (\tau+1)\sqrt{40/p}$. When $i \neq j$, with probability at least $9/10$, it holds that
  \begin{align}
    |v_i^{\dag}v_j| = \frac{|u_i^{\dag}SS^{\dag}u_j|}{\sigma_i\sigma_j} \leq \frac{|u_i^{\dag}(SS^{\dag}-WW^{\dag})u_j|}{\sigma_i\sigma_j} \leq \frac{\theta\sum_{\ell=1}^{\tau}\pnorm{F}{S_{\ell}}^2}{\gamma\sum_{\ell=1}^{\tau}\pnorm{F}{W_{\ell}}^2} \leq \frac{(\tau+1)\theta}{\gamma},
  \end{align}
  where the second inequality follows from \cref{cor:aass}, and the last inequality uses \cref{claim:asa}. Similarly, when $i = j$, the following holds with probability at least $9/10$.
  \begin{align}
    |\norm{v_i}-1| = \frac{|u_i^{\dag}SS^{\dag}u_i-\sigma_i^2|}{\sigma_i^2} \leq \frac{|u_i^{\dag}(SS^{\dag}-WW^{\dag})u_i|}{\sigma_i^2} \leq \frac{(\tau+1)\theta}{\gamma}.
  \end{align}
  Since $|(V^{\dag}V)(i,j)| = |v_i^{\dag}v_j|$, each diagonal entry of $V^{\dag}V$ is at most $(\tau+1)\theta/\gamma$ away from $1$ and each off-diagonal entry is at most $(\tau+1)\theta/\gamma$ away from $0$. More precisely, let $M \in \bbc^{n \times n}$ be the matrix with all ones, i.e., $M(i,j) = 1$ for all $i, j \in \{1, \ldots, n\}$, then for all $i, j \in \{1, \ldots, n\}$, we have
  \begin{align}
    \left(I - \frac{(\tau+1)\theta}{\gamma} M\right)(i, j) \leq (V^{\dag}V)(i, j) \leq \left(I + \frac{(\tau+1)\theta}{\gamma} M\right)(i, j).
  \end{align}

  To prove statement 1, we consider the QR decomposition of $V$. Let $Q \in \bbc^{n \times n}$ be a unitary and $R \in \bbc^{n \times k}$ be upper-triangular with positive diagonal entries satisfying $V = QR$. Since $V^{\dag}V = R^{\dag}R$, we have
  \begin{align}
    \left(I - \frac{(\tau+1)\theta}{\gamma} M\right)(i, j) \leq (R^{\dag}R)(i, j) \leq \left(I + \frac{(\tau+1)\theta}{\gamma} M\right)(i, j).
  \end{align}
  Let $\widehat{R}$ be the upper $k \times k$ part of $R$. Since $R$ is upper-triangular, $\widehat{R}^{\dag}\widehat{R} = R^{\dag}R$. Hence, $\widehat{R}$ can be viewed as an approximate Cholesky factorization of $I$ with error $\frac{(\tau+1)\theta}{\gamma}$. As a consequence of \cite[Theorem 1]{CPS96}, we have $\pnorm{F}{R - I} \leq \frac{\tilde{r}(\tau+1)\theta}{\sqrt{2}\gamma} + O((\tilde{r}(\tau+1)\theta/\gamma)^2)$. Now, we define $R' \in \bbc^{n \times k}$ as the matrix with $I$ on the upper $k \times k$ part and zeros everywhere else. Let $U = QR'$. Clearly, $U$ is isometry as it contains the first $k$ columns of $Q$. To see the column vectors of $V$ span the column space of $V$, note that $UU^{\dag} = QR'R'^{\dag}Q^{\dag}$ where $R'R'^{\dag}$ only contains $I$ on its upper-left $(k \times k)$-block, and $VV^{\dag} = QRR^{\dag}Q^{\dag}$ where $R^{\dag}R$ only contains a diagonal matrix on its upper-left $(k \times k)$-block. To bound the distance between $U$ and $V$, we have $\pnorm{F}{U - V} = \pnorm{F}{Q(R' - R)} = \pnorm{F}{R'-R} \leq \frac{\tilde{r}(\tau+1)\theta}{\sqrt{2}\gamma} + O((\tilde{r}(\tau+1)\theta/\gamma)^2)$. Then statement follows from the choice of $p$ and $\gamma$ in \cref{step:parameter} of \cref{alg:approx_v} and the fact that $\tilde{r} \leq \tau r$.

  Statement 2 follows from the triangle inequality:
  \begin{align}
    \norm{V} - 1 &= \norm{V} - \norm{U} \leq \norm{V - U} \leq \pnorm{F}{V - U}, \text{ and }\\
    1 - \norm{V} &= \norm{U} - \norm{V} \leq \norm{U - V} \leq \pnorm{F}{U - V}.
  \end{align}

  For statement 3, we have
  \begin{align}
    \pnorm{F}{VV^{\dag} - \Pi_V} &= \pnorm{F}{VV^{\dag} - UU^{\dag}} \\
                                 &\leq \pnorm{F}{V(V^{\dag} - U^{\dag})} + \pnorm{F}{(V-U)U^{\dag}} \\
                                 &\leq \norm{V}\pnorm{F}{V^{\dag} - U^{\dag}} + \pnorm{F}{V-U}\norm{U^{\dag}} \\
                                 &\leq \frac{\sqrt{2}\tilde{r}(\tau+1)\theta}{\gamma} + O\left(\left(\frac{\sqrt{2}\tilde{r}(\tau+1)\theta}{\gamma}\right)^2\right),
  \end{align}
  and statement 3 follows from the choice of $p$ and $\gamma$ and the fact that $\tilde{r} \leq \tau r$ in \cref{step:parameter} of \cref{alg:approx_v}.

  Similarly, statement 4 follows by bounding the distance $\pnorm{F}{V^{\dag}V - U^{\dag}U}$.
\end{proof}

Note that the analysis in the proof gives a tighter bound of $\frac{\sqrt{2}\epsilon}{\tau^3r^2} + O(\epsilon^2)$. However, for the convenience of the analysis in the rest of the paper, we choose a looser bound $\ev$ as in \cref{lemma:almost-orthonormal}.

\cref{alg:approx_v} is similar to the main algorithm in~\cite{FKV04} except for the different sampling method used here. In terms of the low-rank approximation, a similar result holds as follows.

\begin{lemma}
  \label{lemma:modfkv}
  Let $A = A_1 + \cdots + A_{\tau} \in \bbc^{n \times n}$ be a Hermitian matrix where $A_{\ell} \in \bbc^{n \times n}$ is Hermitian, $\norm{A_{\ell}} \leq 1$, and $\rank(A_{\ell}) \leq r$ for all $\ell \in [\tau]$. The sampling access each $A_{\ell}$ is given as in \cref{defn:sampling-informal}. Take $A$ and error parameter $\epsilon$ as the input of \cref{alg:approx_v} to obtain the $\sigma_1, \ldots, \sigma_{\tilde{r}}$ and $u_1, \ldots, u_{\tilde{r}}$. Let $V \in \bbc^{n \times \tilde{r}}$ be the matrix such that $V(\cdot, j) = \frac{S^{\dag}}{\sigma_j}u_j$ for $j \in \{1, \ldots, \tilde{r}\}$. Then, with probability at least $9/10$, it holds that $\pnorm{F}{AVV^{\dag} - A} \leq \ef$.
\end{lemma}

To begin with, we first prove the following lemma, which is adapted from~\cite{FKV04}.

\begin{lemma}
  \label{lem:a-ayy}
  Let $A = A_1 + \cdots + A_{\tau} \in \bbc^{n \times n}$ be a Hermitian matrix where $A_{\ell} \in \bbc^{n \times n}$ is Hermitian and $\rank(A_{\ell}) \leq r$ for $\ell \in [\tau]$. The sampling access each $A_{\ell}$ is given as in \cref{defn:sampling-informal}. Apply \cref{proc:samp_row} on $A$ (to obtain row indices $i_1, \ldots, i_p$) and let $S$ be defined as \cref{eq:s}. Then with probability at least $9/10$, there exists an orthonormal set of vectors $\{y_1, \ldots, y_{\hat{r}'}\}$ in the row space of $S$ such that
  \begin{align}
    \pnorm{F}{A - A\sum_{j=1}^{\hat{r}'}y_iy_i^{\dag}}^2 \leq \frac{10\hat{r}\tau}{p}\sum_{\ell=1}^{\tau}\pnorm{F}{A_{\ell}}^2,
  \end{align}
  where $\hat{r}'$ and $\hat{r}$ satisfy $\hat{r}' \leq \hat{r} \leq \tau r$.
\end{lemma}

\begin{proof}
  Denote the rank of $A$ by $\hat{r}$. First note that $\hat{r} \leq \tau r$. We write the singular value decomposition of $A$ as
  \begin{align}
    A = \sum_{t=1}^{\hat{r}}\sigma_tu_tv_t^{\dag}.
  \end{align}
  For the convenience of analysis, we define $\mathcal{S}$ as the set of row indices $\{i_1, \ldots, i_p\}$ that was sampled from \cref{step:it} of \cref{proc:samp_row}. For $t \in [\hat{r}]$, we define the vector-valued random variable
  \begin{align}
    w_t = \frac{1}{p}\sum_{i \in \mathcal{S}}\frac{u_t(i)}{P_i}A(i, \cdot).
  \end{align}
  Note that $w_t$ can be viewed as the average of $p$ i.i.d.~random variables $X_1, \ldots, X_p$ where each $X_j$ is defined as
  \begin{align}
    X_j = \frac{u_t(i)}{P_i}A(i, \cdot), \quad \text{ with probability }P_i, \quad \text{ for }i \in [n].
  \end{align}
  The expected value of $X_j$ can be computed as
  \begin{align}
    \E(X_j) = \sum_{i=1}^n\frac{u_t(i)}{P_i}A(i, \cdot)P_i = u_t^{\dag}A = \sigma_tv_t^{\dag}.
  \end{align}
  Hence, we have
  \begin{align}
    \E(w_t) = \sigma_tv_t^{\dag}.
  \end{align}
  As in \cref{proc:samp_row}, $P_i = \sum_{\ell=1}^{\tau}\norm{A_{\ell}(i, \cdot)}^2/\sum_{\ell'=1}^{\tau}\pnorm{F}{A_{\ell'}}^2$. It follows that
  \begin{align}
    \E\left(\norm{w_t - \sigma_tv_t^{\dag}}^2\right) &= \frac{1}{p}\sum_{i=1}^{\tau}\frac{|u_t(i)|^2\norm{A(i, \cdot)}^2}{P_i} - \frac{\sigma_t^2}{p} \\
                                          &\leq \frac{1}{p}\sum_{i=1}^n|u_t(i)|^2\norm{\sum_{\ell=1}^mA_{\ell}(i, \cdot)}^2\frac{\sum_{\ell'=1}^{\tau}\pnorm{F}{A_{\ell'}}^2}{\sum_{\ell=1}^{\tau}\norm{A_{\ell}(i, \cdot)}^2} \\
                                          &\leq \frac{1}{p}\sum_{i=1}^n|u_t(i)|^2 \tau \sum_{\ell'=1}^{\tau}\pnorm{F}{A_{\ell'}}^2 \\
    \label{eq:wt-stvt}
    &= \frac{\tau}{p}\sum_{\ell=1}^{\tau}\pnorm{F}{A_{\ell}}^2,
  \end{align}
  where the second inequality follows from the Cauchy–Schwarz inequality, and the last equality follows from the fact that $u_t$ is a unit vector.

  If all the random variables $w_t$ happen to be $\sigma_tv_t^{\dag}$, then $\pnorm{F}{A-A\sum_{j=1}^{\hat{r}}w_tw_t^{\dag}}^2 = 0$. In the following, we bound the distance of $w_t$ from its expected value. For $t \in [\hat{r}]$, define $\hat{y}_t = \frac{1}{\sigma_t}w_t^{\dag}$. Let $\{y_1, \ldots, y_n\}$ be an orthonormal basis of $\bbc^n$ with $\sspan(\hat{y}_1, \ldots, \hat{y}_{\hat{r}}) = \sspan(y_1, \ldots, y_{\hat{r}'})$, where $\hat{r}'$ is the dimension of $\sspan(\hat{y}_1, \ldots, \hat{y}_{\hat{r}})$. Define the matrix $F$ as
  \begin{align}
    F = \sum_{t=1}^{\hat{r}'}Ay_ty_t^{\dag},
  \end{align}
  which will be used to approximate $A$. We also define an matrix $\widehat{F}$ as follows that will be used in the intermediate steps to bound the distance:
  \begin{align}
    \widehat{F} = \sum_{t=1}^{\hat{r}}Av_t\hat{y}_t^{\dag}.
  \end{align}
  We have
  \begin{align}
    \label{eq:a-f}
    \pnorm{F}{A-F}^2 = \sum_{t=1}^n\norm{(A-F)y_t}^2 = \sum_{t=\hat{r}'+1}^n\norm{Ay_t}^2 = \sum_{t=\hat{r}'+1}^n\norm{(A-\widehat{F})y_t}^2 \leq \pnorm{F}{A-\widehat{F}}^2,
  \end{align}
  where the third equality follows from the fact that $\hat{y}_i^{\dag}y_j = 0$ for all $i \leq \hat{r}$ and $j > \hat{r}'$.
  In addition, we have
  \begin{align}
    \pnorm{F}{A - \widehat{F}}^2 = \sum_{t=1}^n\norm{u_t^{\dag}(A - \widehat{F})}^2 = \sum_{t=1}^{\hat{r}}\norm{\sigma_tv_t^{\dag} - w_t}^2.
  \end{align}
  Taking the expected value $\pnorm{F}{A - \widehat{F}}^2$ and using \cref{eq:wt-stvt}, we have
  \begin{align}
    \E(\pnorm{F}{A - \widehat{F}}^2) \leq \frac{\hat{r}\tau}{p}\sum_{\ell=1}^{\tau}\pnorm{F}{A_{\ell}}^2.
  \end{align}
  Therefore,
  \begin{align}
    \Pr\left(\pnorm{F}{A-F} \geq \frac{10\hat{r}\tau}{p}\sum_{\ell=1}^{\tau}\pnorm{F}{A_{\ell}}^2\right) \leq \Pr\left(\pnorm{F}{A-\widehat{F}} \geq \frac{10\hat{r}\tau}{p}\sum_{\ell=1}^{\tau}\pnorm{F}{A_{\ell}}^2\right) \leq \frac{1}{10}.
  \end{align}
\end{proof}

Similarly, from $S$ to $W$ using \cref{proc:samp_col}, we have the following corollary.
\begin{corollary}
  \label{cor:s-syy}
  Let $A = A_1 + \cdots + A_{\tau} \in \bbc^{n \times n}$ be a Hermitian matrix where $A_{\ell} \in \bbc^{n \times n}$ is Hermitian and $\rank(A_{\ell}) \leq r$ for $\ell \in [\tau]$. The sampling access each $A_{\ell}$ is given as in \cref{defn:sampling-informal}. Apply \cref{proc:samp_row} and \cref{proc:samp_col} on $A$ (to obtain row indices $i_1, \ldots, i_p$ and column indices $j_1, \ldots, j_p$) and let $S$ be defined as \cref{eq:s} and $W$ be defined as in \cref{eq:w}. Then with probability at least $9/10$, there exists an orthonormal set of vectors $\{y_1, \ldots, y_{\hat{r}'}\}$ in the row space of $S$ such that
  \begin{align}
    \pnorm{F}{S - S\sum_{j=1}^{\hat{r}'}y_iy_i^{\dag}}^2 \leq \frac{10\hat{r}\tau(\tau+1)}{p}\sum_{\ell=1}^{\tau}\pnorm{F}{S_{\ell}}^2,
  \end{align}
  where $\hat{r}$ and $\hat{r}'$ satisfy $\hat{r}' \leq \hat{r} \leq \tau r$.
\end{corollary}
The proof of this corollary is similar to that of \cref{lem:a-ayy}, except that the sampling probability satisfies
\begin{align}
  P_j' \geq \frac{\sum_{\ell'=1}^{\tau}\norm{S_{\ell'}(\cdot, j)}^2}{(\tau+1)\sum_{\ell=1}^{\tau}\pnorm{F}{S_{\ell}}^2}.
\end{align}

Now, we introduce a new notation which is defined in~\cite{FKV04} for the proofs in this section. For a matrix $M$ and a set of vectors $x_i$, $i\in \mathcal{I}$.
\begin{eqnarray}
  \Delta(M;x_i,i \in \mathcal{I}) := \|M\|^2_F - \pnorm{F}{M-M\sum_{i\in \mathcal{I}}x_ix_i^{\dag}}^2 .
\end{eqnarray}
Note that when $x_i$ forms a set of orthogonal unit vectors,
\begin{eqnarray}
  \Delta(M;x_i,i \in \mathcal{I}) &=& \sum_{i\in \mathcal{I}} x_i^{\dag} M^{\dag}M x_i.
\end{eqnarray}

The following lemma adapted from~\cite[Lemma 3]{FKV04} will be useful for our analysis.
\begin{lemma}
  \label{lem:atos}
  Let $A = \sum_{\ell=1}^{\tau} A_\ell$ where $A_\ell\in \mathbb{C}^{n\times n}$ for $\ell\in [\tau]$. Let $A$ and $S\in \mathbb{C}^{k\times n}$ be matrices with same number of columns, and $\|A^{\dag}A - S^{\dag}S\|_F\leq \theta \sum_{\ell=1}^{\tau}\|A_\ell\|^2_F$. Then,
\begin{enumerate}
    \item For any unit vectors $z$ and $z'$ in the row space of $A$,
\begin{align}
      |z^{\dag}A^{\dag}Az' - z^{\dag}S^{\dag}Sz'|\leq \theta\sum_\ell \|A_\ell\|_F^2.
\end{align}
    \item For any set of unit vectors $z_1,\dots,z_h$ in the row space of $A$ and $h\leq k$,
\begin{align}
      |\Delta(A;z_i,i\in [h])-\Delta(S;z_i,i\in [h])|\leq k^2\theta\left(\sum_{\ell=1}^{\tau} \|A_\ell\|_F^2\right).
\end{align}
\end{enumerate}
\end{lemma}

\begin{proof}
The first part is true by following the submultiplicity of matrix norms.
\begin{align}
  |zA^{\dag}Az' - zS^{\dag}Sz'| &= |z(A^{\dag}A - S^{\dag}S)z'|\\
                                &\leq \norm{z}\|A^{\dag}A - S^{\dag}S\|\\
                                &\leq \|A^{\dag}A - S^{\dag}S\| \leq \theta\sum_{\ell=1}^{\tau}  \|A_\ell\|^2_F.
\end{align}

For the second part of the lemma, we see that
\begin{align}
  \Delta(A;z_i,i\in [h]) &= \|A\|^2_F - \|A-A\sum_iz_iz_i^{\dag}\|_F^2\\
                         &= \Tr(AA^{\dag}) -\Tr\left((A-A\sum_iz_iz_i^{\dag})(A^{\dag}-\sum_iz_iz_i^{\dag}A^{\dag})\right)\\
                         &= 2\Tr\left(A(\sum_iz_iz_i^{\dag})A^{\dag}\right) - \Tr\left(A(\sum_iz_iz_i^{\dag})(\sum_iz_iz_i^{\dag})A^{\dag}\right)\\
                         &= 2\sum_{i}z_i^{\dag}A^{\dag}Az_{i} - \sum_{i}z_i^{\dag}A^{\dag}Az_{i} -\sum_{i\neq i'} (z_i^{\dag}z_{i'})z_{i'}^{\dag}A^{\dag}Az_i\\
                         &=\sum_{i}z_i^{\dag}A^{\dag}Az_{i} -\sum_{i\neq i'} (z_i^{\dag}z_{i'})z_{i'}^{\dag}A^{\dag}Az_i.
\end{align}
Similarly, we have
\begin{align}\label{eq:referred-by-eq153}
\Delta(S;z_i,i\in [h]) = \sum_{i}z_i^{\dag}S^{\dag}Sz_{i} -\sum_{i\neq i'} (z_i^{\dag}z_{i'})z_{i'}^{\dag}S^{\dag}Sz_i.
\end{align}
Then, by applying the first part of the lemma,
\begin{align}
|\Delta(A;z_i,i\in [h])-\Delta(S;z_i,i\in [h])| &=
|(\sum_{i}z_i^{\dag}A^{\dag}Az_{i} -\sum_{i\neq i'} (z_i^{\dag}z_{i'})z_{i'}^{\dag}A^{\dag}Az_i) \nonumber\\
&\quad- (\sum_{i}z_i^{\dag}S^{\dag}Sz_{i} -\sum_{i\neq i'} (z_i^{\dag}z_{i'})z_{i'}^{\dag}S^{\dag}Sz_i)|\\
&\leq |\sum_{i}z_i^{\dag}A^{\dag}Az_{i} - \sum_{i}z_i^{\dag}S^{\dag}Sz_{i}| \nonumber\\
&\quad-|\sum_{i\neq i'} (z_i^{\dag}z_{i'})z_{i'}^{\dag}A^{\dag}Az_i - \sum_{i\neq i'} (z_i^{\dag}z_{i'})z_{i'}^{\dag}S^{\dag}Sz_i|\\
&\leq k\theta \sum_\ell \|A_\ell\|_F^2 + (k^2-k) \theta \sum_\ell \|A_\ell\|_F^2 \\
&= k^2\theta \sum_\ell \|A_\ell\|_F^2.
\end{align}
\end{proof}

Now, we are ready to prove \cref{lemma:modfkv}.
\begin{proof}[Proof of \cref{lemma:modfkv}]
  In this proof, we assume that \cref{eq:aass,eq:ssww} hold. (Note that \cref{cor:aass} states that these hold with high probability.)

  We first show two facts that will be used later. The first fact states that the vectors $v_t$ are almost unit vectors:
\\\\
  \noindent
  \textbf{Fact 1.}
  \begin{align}
    \label{eq:fact1}
    \norm{v_t}^2 \leq 1 + \frac{\theta(\tau+1)}{\gamma}.
  \end{align}
  To see this, observe first that the first part of \cref{lem:atos} implies
  \begin{align}
  \norm{S^\dag u_t}^2 -\norm{W^\dag u_t}^2 \leq \theta \sumnormF{S}.
  \end{align}
  Therefore,
  \begin{align} \label{eq:normofv}
  \left|\frac{\norm{S^\dag u_t}2}{\norm{W^\dag u_t}^2}  -1 \right| \leq \frac{\theta \sumnormF{S} }{ \gamma \sumnormF{W} } \leq \frac{\theta (\tau+1)}{\gamma}.
  \end{align}
  by  the bound in \cref{alg:approx_v} and \cref{claim:asa}. \cref{eq:fact1} follows immediately.
\\\\
  \noindent
  \textbf{Fact 2.}
  \begin{align}
    \label{eq:fact2}
    \Delta(S;v_t,t\in T)\geq \Delta(S;u_t,t\in T) - \L(\frac{\theta(\tau+1)}{\gamma}+ \frac{6\tau(\tau+1)^2 \theta^2}{\gamma^2}\R) \sum_\ell \pnorm{F}{A_\ell}^2.
  \end{align}
  To show this, observe that by \cref{cor:aass} we have
  \begin{align} \label{eq:ssss-wwww}
    \pnorm{F}{SS^\dag S S^\dag -W W^\dag W W^\dag}  &\leq \normF{SS^\dag (S S^\dag -W W^\dag)} +\normF{(S S^\dag -W W^\dag) W W^\dag} \\
&\leq \theta \sum_\ell \normF{S_\ell}^2 (\normF{S}^2 +\normF{W}^2).
  \end{align}
  and that for $T\neq  T'\in T$,
  \begin{align} \label{eq:uwwu}
    u_t^\dag W W^\dag u_{t'}^\dag =  u_t^\dag W W^\dag W W^\dag u_{t'}^\dag =0.
  \end{align}
  Now consider  $T\neq  T'\in T$. Recall that $v_t=\frac{S^\dag u_t}{\norm{W^\dag u_t}}$, we have
  \begin{align} \label{eq:claim_cross_term}
    (v_t^\dag v_{t'} ) (v_t^\dag S^\dag S v_{t'} ) =\frac{ (u_t^\dag S S^\dag u_{t'} )(u_t^\dag S S^\dag S S^\dag u_{t'} )}{\norm{W^\dag u_t}^2\norm{W^\dag u_{t'}}^2}.
  \end{align}
  This is the cross term to be used in the calculation of \cref{eq:fact2}. We now bounds its norm. First note that
  \begin{align} \label{eq:ussu}
    \norm{u_t^\dag S S^\dag u_{t'}} = \norm{u_t^\dag (S S^\dag -W W^\dag )u_{t'}} \leq \theta \sum_\ell \normF{S_\ell}^2,
  \end{align}
  where we used \cref{eq:uwwu} in the equality. Similarly, using \cref{eq:ssss-wwww},
  \begin{align} \label{eq:ussssu}
    \norm{u_t^\dag S S^\dag S S^\dag u_{t'} }=\norm{u_t^\dag (S S^\dag S S^\dag -W W^\dag W W^\dag) u_{t'} } \leq \theta \sum_\ell \normF{S_\ell}^2 (\normF{S}^2 +\normF{W}^2).
  \end{align}
  Using \cref{eq:ussu}, \cref{eq:ussssu}, and  the bound in \cref{alg:approx_v} in \cref{eq:claim_cross_term},
  \begin{align}
    \norm{(v_t^\dag v_{t'} ) (v_t^\dag S^\dag S v_{t'} )} &\leq \frac{\theta^2 \L(\sum_\ell \normF{S_\ell}^2 \R)^2 \L(\normF{S}^2 +\normF{W}^2 \R) }{\gamma^2 \sum_\ell (\normF{W_\ell}^2)^2} \\
                                                          &\leq \frac{\theta^2 \L(\sum_\ell \normF{S_\ell}^2 \R)^2 \L(\tau \sum_\ell \normF{S_\ell}^2 +\tau \sum_\ell \normF{W_\ell}^2 \R) }{\gamma^2 \sum_\ell (\normF{W_\ell}^2)^2} \label{eq:claim1corss-2} \\
                                                          &\leq \frac{\tau \theta^2 \L(\sum_\ell \normF{S_\ell}^2 \R)^2 \L( \sum_\ell \normF{S_\ell}^2 + 2 \sum_\ell \normF{S_\ell}^2 \R) }{\gamma^2 \sum_\ell \L((\frac{1}{\tau+1} \normF{S_\ell}^2\R)^2} \label{eq:claim1corss-3} \\
                                                          &\leq \frac{6\tau(\tau+1)^2 \theta^2}{\gamma^2} \sumnormF{A},  \label{eq:claim1corss-4}
  \end{align}
  where we used Cauchy-Schwarz inequality in \cref{eq:claim1corss-2} and \cref{claim:asa} in \cref{eq:claim1corss-3,eq:claim1corss-4}, respectively. Next, we bound the diagonal terms.  By Chauchy-Schwarz inequality
  \begin{align}
    \norm{u} \norm{S S^\dag u} \geq u^\dag S S^\dag u = \norm{S^\dag u}^2.
  \end{align}
  So for $t\in [\tilde{r}]$,
  \begin{align}
    v_t^\dag S^\dag S v_t = \frac{u_t^\dag S S^\dag S S^\dag u_t}{\norm{W^\dag u_t}^2} \geq \frac{\norm{S^\dag u_t}^4}{\norm{W^\dag u_t}^2}.
  \end{align}
  We then have
  \begin{align}
    \sum_{t\in [\tilde{r}]} v_t^\dag S^\dag S v_t &\geq \sum_{t \in [\tilde{r}]} \frac{\norm{S^\dag u_t}^4}{\norm{W^\dag u_t}^2} \\
                                                 &= \sum_{t \in [\tilde{r}]} (u_t^\dag S^\dag S u_t )\L( \frac{\norm{S^\dag u_t}^2}{\norm{W^\dag u_t}^2} \R) \\
                                                 &\geq \L(1 -\frac{\theta(\tau+1)}{\gamma}\R)  \sum_{t \in [\tilde{r}]} u_t\ct S S\ct u_t \label{eq:claim1diag-3} \\
                                                 &= \L(1 -\frac{\theta(\tau+1)}{\gamma}\R) \Delta(S\ct; u_t, t\in [\tilde{r}]), \label{eq:claim1diag}
  \end{align}
  where we used \cref{eq:normofv} in \cref{eq:claim1diag-3}. Putting \cref{eq:claim1corss-4,eq:claim1diag} into \cref{eq:referred-by-eq153} of \cref{lem:atos},
  \begin{align}
    \Delta(S;v_t,t\in [\tilde{r}]) &=\sum_{t \in [\tilde{r}]} v_t\ct S\ct S v_t -\sum_{t\neq t'} (v_t\ct c_{t'}) v_{t'}\ct S\ct S v_t \\
                         &\geq \L(1 -\frac{\theta(\tau+1)}{\gamma} \R) \Delta(S\ct;u_t,t\in [\tilde{r}]) - \frac{6\tau(\tau+1)^2 \theta^2\tilde{r}^2}{\gamma^2} \sumnormF{A} \\
                         & \geq \Delta(S\ct;u_t,t\in [\tilde{r}]) -\L(\frac{2\tau\theta(\tau+1)}{\gamma}+ \frac{6\tau(\tau+1)^2 \theta^2\tilde{r}^2}{\gamma^2}\R)\sumnormF{A},
  \end{align}
  where we used $\Delta(S\ct;u_t,t\in [\tilde{r}]) \leq \normF{S}^2 \leq 2\tau \sumnormF{A}$ in the last line. This completes the proof of \cref{eq:fact2}.

  By \cref{lemma:modfkv}, with probability at least $9/10$, there exist orthonormal vectors $x_t$ for $t \in [\hat{r}']$ (for $\hat{r}' \leq \hat{r}$) in the row space of $S$ such that
  \begin{align}
    \Delta(A; x_t, t \in [\hat{r}']) \geq \pnorm{F}{A}^2 - \pnorm{F}{A - A\sum_{t=1}^{\hat{r}'}x_tx_t^{\dag}}^2 \geq \pnorm{F}{A}^2 - \frac{10\hat{r}\tau}{p}\sum_{\ell=1}^{\tau}\pnorm{F}{A_{\ell}}^2.
  \end{align}
  According to the statement 2 of \cref{lem:atos}, we have
  \begin{align}
    \Delta(S; x_t, t \in [\hat{r}']) \geq \Delta(A; x_t, t \in [\hat{r}']) - \hat{r}'^2\theta\sum_{\ell=1}^{\tau}\pnorm{F}{A_{\ell}}^2 \geq \pnorm{F}{A}^2 - \left(\frac{10\hat{r}\tau}{p} + \hat{r}'^2\theta\right)\sum_{\ell=1}^{\tau}\pnorm{F}{A_{\ell}}^2.
  \end{align}
  Since $S$ and $S^{\dag}$ have the same singular values, there exist orthonormal vectors $y_t$ for $t \in [\hat{r}']$ in the row space of $S^{\dag}$ satisfying
  \begin{align}
    \Delta(S^{\dag}; y_t, t \in [\hat{r}']) \geq \pnorm{F}{A}^2 - \left(\frac{10\hat{r}\tau}{p} + \hat{r}'^2\theta\right)\sum_{\ell=1}^{\tau}\pnorm{F}{A_{\ell}}^2.
  \end{align}
  Now, applying \cref{cor:s-syy}, it holds that with probability at least $9/10$, there exist orthonormal vectors $z_t$ for $t \in [\hat{r}']$ in the row space of $W^{\dag}$ such that
  \begin{align}
    \Delta(S^{\dag}; z_t, t \in [\hat{r}']) &= \pnorm{F}{S}^2 - \pnorm{F}{S - S\sum_{j=1}^{\hat{r}'}z_jz_j^{\dag}}^2 \\
                                            &\geq \pnorm{F}{S}^2 - \frac{10\hat{r}\tau(\tau+1)}{p}\sum_{\ell=1}^{\tau}\pnorm{F}{S_{\ell}}^2 \\
                                            &\geq \pnorm{F}{S}^2 - \pnorm{F}{S - S\sum_{j=1}^{\hat{r}'}y_jy_j^{\dag}}^2 - \frac{10\hat{r}\tau(\tau+1)}{p}\sum_{\ell=1}^{\tau}\pnorm{F}{S_{\ell}}^2 \\
                                            &\geq \Delta(S^{\dag}; y_t, t \in [\hat{r}']) - \frac{10\hat{r}\tau(\tau+1)}{p}\sum_{\ell=1}^{\tau}\pnorm{F}{S_{\ell}}^2 \\
                                            &\geq \pnorm{F}{A}^2 - \left(\frac{20\hat{r}\tau(\tau+1)}{p} + \hat{r}'^2\theta\right)\sum_{\ell=1}^{\tau}\pnorm{F}{A_{\ell}}^2.
  \end{align}
  Again, by the statement 2 of \cref{lem:atos}, we have
  \begin{align}
    \Delta(W^{\dag}; z_t, t \in [\hat{r}']) &\geq \Delta(S^{\dag}; z_t, t\in [\hat{r}']) - \hat{r}'^2\theta\sum_{\ell=1}^{\tau}\pnorm{F}{S_{\ell}}^2 \\
                                            &\geq \pnorm{F}{A}^2 - \left(\frac{20\hat{r}\tau(\tau+1)}{p} + \hat{r}'^2\theta\right)\sum_{\ell=1}^{\tau}\pnorm{F}{A_{\ell}}^2 - \hat{r}'^2\theta\sum_{\ell=1}^{\tau}\pnorm{F}{S_{\ell}}^2 \\
                                            &\geq \pnorm{F}{A}^2 - \left(\frac{20\hat{r}\tau(\tau+1)}{p} + 3\hat{r}'^2\theta\right)\sum_{\ell=1}^{\tau}\pnorm{F}{A_{\ell}}^2.
  \end{align}
  As $u_t$ for $t \in [\hat{r}']$ (computed in \cref{alg:approx_v}) are the left singular vectors of $W$, we have
  \begin{align}
    \Delta(W^{\dag}; u_t, t \in [\hat{r}']) \geq \Delta(W^{\dag}; z_t, t \in [\hat{r}']) \geq \pnorm{F}{A}^2 - \left(\frac{20\hat{r}\tau(\tau+1)}{p} + 3\hat{r}'^2\theta\right)\sum_{\ell=1}^{\tau}\pnorm{F}{A_{\ell}}^2.
  \end{align}

  Recall that $u_t$ for $t \in [\tilde{r}]$ are the singular vectors after the filter in \cref{alg:approx_v} and $u_t$ for $t \in [\hat{r}']$ are orthonormal vectors. It holds that
  \begin{align}
    \Delta(W^{\dag}; u_t, t\in[\tilde{r}]) &\geq \Delta(W^{\dag}; u_t, t\in[\hat{r}']) - \hat{r}\gamma\sum_{\ell=1}^{\tau}\pnorm{F}{W_{\ell}}^2 \\
                                              & \geq \pnorm{F}{A}^2 - \left(\frac{20\hat{r}'\tau(\tau+1)}{p} + 3\hat{r}'^2\theta + 4\hat{r}\gamma\right)\sum_{\ell=1}^{\tau}\pnorm{F}{A_{\ell}}^2.
  \end{align}
  Now, we apply the statement 2 of \cref{lem:atos} one more time. It follows that
  \begin{align}
    \Delta(S^{\dag}; u_t, t \in [\tilde{r}]) &\geq \Delta(W^{\dag}; u_t, t\in[\tilde{r}]) - \tilde{r}^2\theta\sum_{\ell=1}^{\tau}\pnorm{F}{S_{\ell}}^2 \\
                                                & \geq \pnorm{F}{A}^2 - \left(\frac{20\hat{r}\tau(\tau+1)}{p} + 5\hat{r}^2\theta + 4\hat{r}\gamma\right)\sum_{\ell=1}^{\tau}\pnorm{F}{A_{\ell}}^2.
  \end{align}
  \cref{lem:atos} implies that
  \begin{align}
    \Delta(A; v_t, t \in [\tilde{r}]) &\geq \Delta(S; v_t, t \in [\tilde{r}]) - \left(1+ \frac{\theta(\tau+1)}{\gamma}\right)\tilde{r}^2\theta\sum_{\ell=1}^{\tau}\pnorm{F}{A}^2 \\
                                          &\geq \pnorm{F}{A}^2 - \left(\frac{20\hat{r}\tau(\tau+1)}{p} + 5\hat{r}^2\theta + 4\hat{r}\gamma + \left(1+\frac{\theta(\tau+1)}{\gamma}\right)\hat{r}^2\theta+\right.\\
                                          &\left.\left(\frac{2\tau(\tau+1)\theta}{\gamma}+\frac{6\tau(\tau+1)^2\theta^2\hat{r}^2}{\gamma^2}\right)\right)\sum_{\ell=1}^{\tau}\pnorm{F}{A_{\ell}}^2.
  \end{align}
  By the definition of the $\Delta$ function, we conclude that
  \begin{align}
    \pnorm{F}{A - A\sum_{t \in [\tilde{r}]}v_tv_t^{\dag}}^2 &\leq  \left(\frac{20\hat{r}\tau(\tau+1)}{p} + 5\hat{r}^2\theta + 4\hat{r}\gamma + \left(1+\frac{\theta(\tau+1)}{\gamma}\right)\hat{r}\theta \right.\\
                                                               &\left. + \left(\frac{2\tau(\tau+1)\theta}{\gamma}+\frac{6\tau(\tau+1)^2\theta^2\hat{r}^2}{\gamma^2}\right)\right)\sum_{\ell=1}^{\tau}\pnorm{F}{A_{\ell}}^2.
  \end{align}
  The claim bound follows by the choice of $p$ and $\gamma$ in \cref{step:parameter} of \cref{alg:approx_v}, $\theta$ in \cref{cor:aass}, and the fact that $\hat{r} \leq \tau r$.
\end{proof}

Now we are ready to prove \cref{thm:vvavv_bound}.
\thmvvavv*

\begin{proof}
  By \cref{lemma:modfkv}, we have
  \begin{align}
    \pnorm{F}{AVV^{\dag} - A} \leq \frac{\epsilon}{300r^2}.
  \end{align}
  By taking adjoint, we have
  \begin{align}
    \pnorm{F}{VV^{\dag}A - A} \leq \frac{\epsilon}{300r^2}.
  \end{align}
  Then,
  \begin{align}
    \pnorm{F}{VV^{\dag}AVV^{\dag} - A} &\leq \pnorm{F}{VV^{\dag}AVV^{\dag} - AVV^{\dag}} + \pnorm{F}{AVV^{\dag} - A} \\
                                       &\leq \ef\left(1+\ev\right) + \ef,
  \end{align}
  where the last inequality follows from \cref{lemma:almost-orthonormal}. Then the result follows.
\end{proof}


By using \cref{claim:tr_prod}, we approximate $V^\dag AV$ as follows.

\vavb*

\begin{proof}
Let $B_t = V^{\dag}A_tV$ for $t\in [\tau]$ and $B= \sum_{t=1}^{\tau} B_t$. Note that $B_t(i,j) = V^{\dag}(i,\cdot)A_tV(\cdot,j) = \Tr[A_t (V(\cdot,j) \times V^{\dag}(i,\cdot))]$, where $\times$ denotes outer product, and we can get one query access to $(V(\cdot,j) \times V^{\dag}(i,\cdot))$ by using two queries to $V$. There, by \cref{claim:tr_prod}, one can estimate $B_t(i,j)$ with error at most $\epsilon_s/\tilde{r}\tau$ with probability $1-\frac{2\delta}{\tau(\tilde{r}^2+\tilde{r})}$ by using
\[
  O(\|A_t\|_F^2\|V(i,\cdot)\|\|V(j,\cdot)\|\frac{\tilde{r}^2\tau^2}{\epsilon^2_s}\log\frac{(\tilde{r}^2+\tilde{r})\tau}{2\delta} (p+\log n))
\]
queries. We denote the estimation to $B_t(i,j)$ as $\tilde{B}_t(i,j)$.

Since $A_t$ is a Hermitian matrix, we only need to compute $(\tilde{r}^2+\tilde{r})/2$ elements. By union bound over $i,j$, for all $t\in [\tau]$ we have
\begin{align}
    \Pr\L[|B_t(i,j) - \tilde{B}_t(i,j)|\leq \epsilon_s/\tilde{r}\tau\,\,\mbox{for all }\,i,j\in [\tilde{r}]\R]\geq 1-\delta/\tau.
\end{align}
Then, by union bound over $t$, we have
\begin{align}\label{eqn:bt}
    \Pr\L[|B_t(i,j) - \tilde{B}_t(i,j)|\leq \epsilon_s/\tilde{r}\tau\,\mbox{ for all }\,i,j\in [\tilde{r}],\, t\in[\tau]\R]\geq 1-\delta.
\end{align}

 Let $\tilde{B} = \sum_t \tilde{B}_t$. With probability $1-\delta$,
\begin{eqnarray}
    \|B - \tilde{B}\|_F &\leq& \sum_{t=1}^{\tau} \|B_t - \tilde{B}_t\|_F \nn\\
    &\leq&\tau \sqrt{\tilde{r}^2 (\epsilon^2_s/\tilde{r}^2\tau^2)} = \epsilon_s.
\end{eqnarray}
The time complexity is
\begin{align*}
&O(\|A_t\|_F^2\|V(i,\cdot)\|\|V(j,\cdot)\|\frac{\tilde{r}^2\tau^2}{\epsilon^2_s}\left(\log\frac{(\tilde{r}^2+\tilde{r})\tau}{2\delta}\right)(\log n +p) \frac{\tau(\tilde{r}^2+\tilde{r})}{2} \\
&=O\left((p+\log n)\frac{r^{5}\tau^3}{\epsilon_s^2}\log\frac{1}{\delta}\right).
\end{align*}
The equality follows from the facts that $\tilde{r} \leq r$, $\|A_t\|_F \leq \sqrt{r}$ for all $t$ (as $\norm{A_t} \leq 1$), and $\|V(i,\cdot)\|\leq 2$ for all $i$.
\end{proof}

Then, we prove that the matrix multiplication of an isometry and a matrix satisfying \cref{lemma:almost-orthonormal} is still close to an isometry.

\begin{corollary} \label{lem:almost_orth_2}
Let $U\in \mathbb{C}^{r\times \tilde{r}}$ be a unitary matrix and $V\in \mathbb{C}^{n\times \tilde{r}}$ be a matrix which satisfies \cref{lemma:almost-orthonormal} with error parameter $\ev$. Then the following properties hold for the matrix $VU$.
\begin{enumerate}
    \item There exists an isometry $W\in \mathbb{C}^{n\times r}$ such that $W$ spans the column space of $VU$ and $\|VU-W\|_F\leq \ev$.
    \item $|\|VU\|-1|\leq \ev$.
    \item $\|(VU)^{\dag}(VU)-I_r\|_F \leq \ev$.
    \item Let $\Pi_{VU}$ be the projector of the column space of $UV$. Then $\|(VU)(VU)^{\dag}-\Pi_{VU}\|_F\leq \frac{3\epsilon}{300r^2(\tau+1)}$.
\end{enumerate}
\end{corollary}
\begin{proof}
By \cref{lemma:almost-orthonormal}, there exists an isometry $W'\in\mathbb{C}^{n\times r}$ such that $W'$ spans the column space of $V$ and $\|V-W'\|_F\leq \ev$. Let $W=W'U$,
\begin{align}
\|VU-W\|_F=\|VU - W'U\|_F \leq \|V-W'\|_F\|U\| \leq \ev.\label{eq:orth_1}
\end{align}
Note that $W$ is also an isometry.

For the second property, by \cref{eq:orth_1}, we can get the following inequality
\begin{eqnarray}
|\|VU\|-1|=|\|VU\| - \|W\||  \leq \|VU-W\| \leq \ev.
\end{eqnarray}

For the third inequality,
\begin{eqnarray}
\|(VU)^{\dag}(VU) - I\|_F = \|U^{\dag} V^{\dag}V U - U^{\dag}U\|_F \leq \|U^{\dag}\|\| V^{\dag}V - I\|_F \|U^{\dag}\| \leq \ev.
\end{eqnarray}
The last inequality holds because of \cref{lemma:almost-orthonormal}.

Lastly,
\begin{align}
\|VUU^{\dag}V^{\dag} - \Pi_{VU}\|_F &= \|VUU^{\dag}V^{\dag} - WW^{\dag}\|_F\\
&= \|VUU^{\dag}V^{\dag} - VUW' + VUW' - W'UU^{\dag}W'^{\dag}\|_F\\
&\leq \|VU\|\|U^{\dag}V^{\dag} - W'^{\dag}\|_F + \|VU-W'\|_F\|W'^{\dag}\|\\
&\leq \left(1+\ev\right) \ev + \ev\\
&\leq \frac{3\epsilon}{300r^2(\tau+1)}.
\end{align}
\end{proof}

We conclude by giving the proof of \cref{thm:our_appx}.
\thmsvd*
\begin{proof}[Proof of \cref{thm:our_appx}]

By \cref{lem:appx_vav}, we can compute $\tilde{B}$ in time $O((p+\log n) \frac{r^{5}\tau^3}{\epsilon^2}\log\frac{1}{\delta})$. Recall that statements 1 to 4 in \cref{lemma:almost-orthonormal} hold for $V$. We have
\begin{align}
\|V\tilde{B}V^\dag - VV^\dag AVV^\dag +VV^\dag A VV^\dag - A\| &\leq \|V\tilde{B}V^\dag - VV^\dag AVV^\dag\|+\|VV^\dag A VV^\dag - A\| \\
&\leq \left(1+\ev\right)^2\es \nonumber\\
&\ +\,\left(2+\ev\right)\ef.
\end{align}
The first term of the last inequality comes from \cref{lem:appx_vav} with $\epsilon_s = \es$.
The second statement directly follows from \cref{lem:almost_orth_2}.
\end{proof}


\section{Proof of \cref{thm:est_gibbs}}\label{sec:gibbs_proof}
We first prove two technical lemmas.
\begin{lemma}\label{lem:tr_dist}
Let $A\in \mathbb{C}^{n\times n}$, $B\in \mathbb{C}^{n\times n}$, and $B'\in \mathbb{C}^{n\times n}$ be Hermitian matrices. Suppose $\|B-B'\|\leq \ev$. Then
\begin{align}
|\Tr[AB] - \Tr[AB']| \leq \ev\Tr{|A|}.
\end{align}
\end{lemma}

\begin{proof}
Let $A = \sum_i \sigma_i v_iv_i^\dag$. We have
\begin{align}
    |\Tr[AB] - \Tr[AB']|=\sum_i \sigma_i v_i^\dag (B-B') v_i\leq\sum_i |\sigma_i| \|B-B'\|\leq\ev \Tr{|A|}.
\end{align}
\end{proof}

\begin{restatable}{lemma}{fidelity}
\label{lem:fidelity}
Let $A$ and $B$ be Hermitian matrices, and $\|A - B\|\leq \epsilon$. Let $\rho_A(\frac{\epsilon}{2}) = \frac{e^{-\frac{\epsilon}{2} A}}{\Tr{[e^{-A}]}}$ and $\rho_B(\frac{\epsilon}{2}) = \frac{e^{-\frac{\epsilon}{2} B}}{\Tr{[e^{-B}]}}$. Then $F(\rho_A(\frac{\epsilon}{2}),\rho_B(\frac{\epsilon}{2}))\geq e^{-\frac{\epsilon}{2} \epsilon}$, where $F(\rho_{A},\rho_{B}):=\Tr\big[\sqrt{\sqrt{\rho_{A}}\rho_{B}\sqrt{\rho_{A}}}\,\big]$ is the fidelity between $\rho_{A}$ and $\rho_{B}$.
\end{restatable}

\begin{proof}
This lemma has been proven in \cite[Appendix C]{PW09}. We give the proof here for completeness.

Let $\sigma_1(A)>\cdots>\sigma_n(A)$ and $\sigma_1(B)>\cdots>\sigma_n(B)$ be eigenvalues of $A$ and $B$. Then, by the definition of spectral norm,
\begin{align}
    \max_{j\in\range{n}} |\sigma_j(A) - \sigma_j(B)|\leq \|A-B\|.
\end{align}
This fact implies that
\begin{align}
   e^{-\|A-B\|}\leq e^{\sigma_j(A) - \sigma_j(B)}\leq e^{\|A-B\|},
\end{align}
for all $j\in [n]$. Therefore,
\begin{align}\label{eq:mult_bound}
e^{-\|A-B\|}\Tr{[e^A]}\leq e^{B}\leq e^{\|A-B\|}\Tr{[e^A]}.
\end{align}

Now, we define a function
\begin{align}
    f(p) = \Tr{(e^{pD/2} e^{pC} e^{pD/2})^{1/p}}.
\end{align}
It has been shown in~\cite{PW09,FS94} that the function $f$ has two properties: $f$ is an increasing function for Hermitian matrices $C$ and $D$ in $p\in (0,\infty)$, and when $p\rightarrow 0$, $f(p) = \Tr{[e^{C+D}]}$. Hence,
\begin{align}
    \Tr{[e^{C+D}]}\leq \Tr{\big[(e^{pD/2} e^{pC} e^{pD/2})^{1/p}\big]}
\end{align}
holds for $p>0$.

Let $p=2$, $C=-\beta A/2$, and $D=-\beta B/2$, we bound the fidelity as follows:
\begin{eqnarray}
F(\rho_A(\beta),\rho_B(\beta))=\frac{\Tr{(e^{-\beta B/2}e^{-\beta A}e^{-\beta B})^{1/2}}}{\sqrt{\Tr{\rho_A(\beta)}\Tr{\rho_B(\beta)}}}\geq\frac{\Tr{e^{-\beta(A+B)/2}}}{\sqrt{\Tr{\rho_A(\beta)}\Tr{\rho_B(\beta)}}}.
\end{eqnarray}
Note that $\|B - \frac{A+B}{2}\|\leq \epsilon/2$. By applying \cref{eq:mult_bound},
\begin{eqnarray}
\frac{\Tr{e^{-\beta(A+B)/2}}}{\sqrt{\Tr{\rho_A(\beta)}\Tr{\rho_B(\beta)}}}\geq\frac{e^{-\beta\epsilon/2}\Tr{e^{-\beta B}}}{\sqrt{e^{\beta\epsilon}}\Tr{e^{-\beta B}}}\geq e^{-\beta \epsilon}.
\end{eqnarray}

\end{proof}

\cref{lem:fidelity} implies that the trace distance between $\rho_A$ and $\rho_B$ is
\begin{align}
    \frac{1}{2}\Tr{|\rho_A-\rho_B|} \leq \sqrt{1- e^{-2\frac{\epsilon}{2} \epsilon}},
\end{align}
and the spectral distance is
\begin{align}
    \|\rho_A(\epsilon/2) - \rho_B(\epsilon/2)\|\leq 2\sqrt{1- e^{-2\frac{\epsilon}{2} \epsilon}}.
\end{align}

Recall that $\tilde{A} = VV^\dag AVV^\dag$ and $U$, $D$ are the outputs of \cref{alg:appx_VVAVV}, which will be used in \cref{alg:appx_tr}. In this section, we suppose $\|\tilde{A}-A\|\leq (2+\ev)\ef$ as in \cref{thm:vvavv_bound}.

\estgibbs*
\begin{proof}

  As we have proven in \cref{lem:almost_orth_2}, there exists an isometry $\tilde{U}$ such that $\|\tilde{U}-(VU)\|\leq \ev$ and $\tilde{U}$ spans the column space of $(VU)$. We define two additional Gibbs states $\rho' = \frac{\tilde{U}e^{-\frac{\epsilon}{2} D}\tilde{U}^\dag}{\Tr{e^{-\frac{\epsilon}{2} D}}}$ and $\tilde{\rho} = \frac{e^{-\frac{\epsilon}{2} \tilde{A}}}{\Tr{e^{-\frac{\epsilon}{2} \tilde{A}}}}$.
\begin{align}
    &|\Tr{[A_\ell\rho]}-\zeta|\\ = &|\Tr{[A_\ell\rho]}-\Tr{[A_\ell\tilde{\rho}]} + \Tr{[A_\ell\tilde{\rho}]} - \Tr{[A_\ell\rho']} + \Tr{[A_\ell\rho']} -\Tr{[A_\ell\hat{\rho}]}+
    \Tr{[A_\ell\hat{\rho}]}-\zeta|\\
  \leq &|\Tr{[A_\ell\rho]}-\Tr{[A_\ell\tilde{\rho}]}| + |\Tr{[A_\ell\tilde{\rho}]} - \Tr{[A_\ell\rho']}| + |\Tr{[A_\ell\rho']} -\Tr{[A_\ell\hat{\rho}]}|+
    |\Tr{[A_\ell\hat{\rho}]}-\zeta|.
\end{align}
We give bounds on each term as follows. First,
\begin{align}
|\Tr{[A_\ell\rho]}-\Tr{[A_\ell\tilde{\rho}]}| &\leq \Tr{|A_\ell|} \|\rho - \tilde{\rho}\|\\
&\leq 2\Tr{|A_\ell|}\sqrt{1- e^{-2\frac{\epsilon}{2} ((2+\ev)\ef)}}.\label{eq:err_1}
\end{align}

For $|\Tr{[A_\ell\tilde{\rho}]} - \Tr{[A_\ell\rho']}|$, we first compute an upper bound on $\|(VU)D(VU)^{\dag} - \tilde{U}D\tilde{U}^{\dag}\|$:
\begin{align}
  \|(VU)D(VU)^\dag - \tilde{U}D\tilde{U}^\dag\| \leq \|\tilde{U}-(VU)\| \|D\| (\|VU\|+\|\tilde{U}\|) \leq 3\ev \|D\|.
\end{align}
Then, by applying \cref{lem:fidelity,lem:tr_dist} again, we get
\begin{align}
|\Tr{[A_\ell\tilde{\rho}]} - \Tr{[A_\ell\rho']}|\leq \Tr{|A_\ell|} \|\tilde{\rho} - \rho'\|\leq \Tr{|A_\ell|}\left(2\sqrt{1- e^{-6\frac{\epsilon}{2} \ev\|D\|}}\right). \label{eq:err_2}
\end{align}

For the second last term $|\Tr{[A_\ell\rho']} -\Tr{[A_\ell\hat{\rho}]}|$, it is not hard to show that
\begin{align}
  \|\tilde{U}e^{-\frac{\epsilon}{2} D}\tilde{U}^\dag - (VU)e^{-\frac{\epsilon}{2} D}(VU)^\dag\| \leq \left(2+\ev\right) \|\tilde{U}-VU\|\Tr{[e^{-\frac{\epsilon}{2} D}]}.
\end{align}
Then,
\begin{align}
  |\Tr{[A_\ell\rho']} -\Tr{[A_\ell\hat{\rho}]}| &\leq \Tr{|A_\ell|} \|\rho' - \hat{\rho}\|\leq \left(2+\ev\right)\|\tilde{U}-VU\|\Tr{|A_\ell|} \nonumber \\
                                                &\leq 3\ev \Tr{|A_\ell|}. \label{eq:err_3}
\end{align}
The last term follows from \cref{claim:tr_prod} by setting the precision to be $\epsilon/5$. Hence
\begin{eqnarray}
    |\Tr{[A_\ell\hat{\rho}]}-\zeta|\leq \epsilon/5.
\end{eqnarray}

By adding \cref{eq:err_1,eq:err_2,eq:err_3} together,
\begin{eqnarray}
    |\Tr{[A_\ell\rho]}-\zeta|\leq\epsilon.
\end{eqnarray}
Let $Q(\cdot)$ ($S(\cdot)$, and $N(\cdot)$, respectively) denote the time complexity for querying to (sampling from, and obtaining the row norms of, respectively) a matrix.
$\Tr{[A_\ell\hat{\rho}]}$ can be approximated with precision $\epsilon/5$ with probability $1-\delta$ in time
\begin{eqnarray}
O\left(\frac{4}{\epsilon^2}(Q(A_\ell)+Q(VU)+S(A_\ell)+N(A_\ell))\log\frac{1}{\delta}\right) = O\left(\frac{4}{\epsilon^2}(\log n+\tau pr)\log\frac{1}{\delta}\right),
\end{eqnarray}
where $p$ is the number of rows sampled in~\cref{alg:approx_v} and the maximum rank of the Gibbs state is $\tau r$. The last equality is true since one can compute $(VU)(i,j)$ by computing $V(i,j)$ as $(S^\dag(i,\cdot)u_j/\sigma_j$ and then compute the inner product $V(i,\cdot)U(\cdot,j)$, which takes $O(p\tau r)$ time.
\end{proof}

\section{Proof of \cref{cor:shadow}}\label{sec:proof-shadow}

\shadowtomography*

\begin{proof}
We denote $A_{i}=E_{i}$ for all $i\in\range{m}$ and $A_{i}=-E_{i-m}$ for all $i\in\{m+1,\ldots,2m\}$; also denote $a_{i}=p_{i}$ for all $i\in\range{m}$ and $a_{i}=-p_{i-m}$ for all $i\in\{m+1,\ldots,2m\}$. As a result, $\Tr[\sigma E_{i}]-p_{i}\leq\epsilon$ is equivalent to $\Tr[\sigma A_{i}]\leq a_{i}+\epsilon$ for all $i\in\range{m}$, and $\Tr[\sigma E_{i}]-p_{i}\geq-\epsilon$ is equivalent to $\Tr[\sigma A_{i+m}]\leq a_{i+m}+\epsilon$ for all $i\in\range{m}$; therefore, the shadow tomography problem in \cref{eqn:shadow-formal} is equivalent to the following SDP feasibility problem:
\begin{align}
\text{Find $\sigma$ such that}\qquad &\Tr[A_{i}\sigma]\leq a_{i}+\epsilon\quad\forall\,i\in\range{2m};  \\
&\sigma\succeq 0,\ \ \Tr[\sigma]=1.
\end{align}
Consequently, \cref{cor:shadow} reduces to the SDP in \cref{eqn:SDP-formal} with $2m$ constraints; the result hence follows from \cref{thm:main}.
\end{proof}


\begin{thebibliography}{10}
\bibitem{aaronson2017quantum}
Scott Aaronson, \emph{Shadow tomography of quantum states}, Proceedings of the
  50th Annual ACM Symposium on Theory of Computing (STOC 2018), pp.~325--338,
  ACM, 2018, \mbox{\href{https://arxiv.org/abs/1711.01053}{arXiv:1711.01053}}.

\bibitem{allen-zhu2016SDP}
Zeyuan Allen-Zhu, Yin~Tat Lee, and Lorenzo Orecchia, \emph{Using optimization
  to obtain a width-independent, parallel, simpler, and faster positive sdp
  solver}, Proceedings of the 27th Annual ACM-SIAM Symposium on Discrete
  Algorithms (SODA 2016), pp.~1824--1831, Society for Industrial and Applied
  Mathematics, 2016,
  \mbox{\href{https://arxiv.org/abs/1507.02259}{arXiv:1507.02259}}.

\bibitem{anstreicher2000volumetric}
Kurt~M. Anstreicher, \emph{The volumetric barrier for semidefinite
  programming}, Mathematics of Operations Research \textbf{25} (2000), no.~3,
  365--380.

\bibitem{vanApeldoorn2018SDP}
Joran~van Apeldoorn and Andr{\'a}s Gily{\'e}n, \emph{Improvements in quantum
  {SDP}-solving with applications}, Proceedings of 46th International
  Colloquium on Automata, Languages, and Programming (ICALP 2019), Leibniz
  International Proceedings in Informatics, vol. 132, pp.~99:1--99:15, 2019,
  \mbox{\href{https://arxiv.org/abs/1804.05058}{arXiv:1804.05058}}.

\bibitem{vanApeldoorn2017quantum}
Joran~van Apeldoorn, Andr{\'a}s Gily{\'e}n, Sander Gribling, and Ronald~{de}
  Wolf, \emph{Quantum {S}{D}{P}-solvers: Better upper and lower bounds},
  Proceedings of the 58th Annual IEEE Symposium on Foundations of Computer
  Science (FOCS 2017), IEEE, 2017,
  \mbox{\href{https://arxiv.org/abs/1705.01843}{arXiv:1705.01843}}.

\bibitem{arora2012survey}
Sanjeev Arora, Elad Hazan, and Satyen Kale, \emph{The multiplicative weights
  update method: a meta-algorithm and applications}, Theory of Computing
  \textbf{8} (2012), no.~1, 121--164.

\bibitem{arora2007combinatorial}
Sanjeev Arora and Satyen Kale, \emph{A combinatorial, primal-dual approach to
  semidefinite programs}, Proceedings of the 39th Annual ACM Symposium on
  Theory of Computing (STOC 2007), pp.~227--236, ACM, 2007.

\bibitem{arrazola2019quantum}
Juan~Miguel Arrazola, Alain Delgado, Bhaskar~Roy Bardhan, and Seth Lloyd,
  \emph{Quantum-inspired algorithms in practice}, 2019,
  \mbox{\href{https://arxiv.org/abs/1905.10415}{arXiv:1905.10415}}.

\bibitem{brandao2018SDP}
Fernando G. S.~L. Brand{\~a}o, Amir Kalev, Tongyang Li, Cedric Yen-Yu Lin,
  Krysta~M. Svore, and Xiaodi Wu, \emph{Quantum {SDP} solvers: Large speed-ups,
  optimality, and applications to quantum learning}, Proceedings of the 46th
  International Colloquium on Automata, Languages, and Programming (ICALP
  2019), vol. 132, p.~27, Schloss Dagstuhl--Leibniz-Zentrum fuer Informatik,
  2019, \mbox{\href{https://arxiv.org/abs/1710.02581}{arXiv:1710.02581}}.

\bibitem{brandao2016quantum}
Fernando G. S.~L. Brand{\~a}o and Krysta Svore, \emph{Quantum speed-ups for
  semidefinite programming}, Proceedings of the 58th Annual IEEE Symposium on
  Foundations of Computer Science (FOCS 2017), IEEE, 2017,
  \mbox{\href{https://arxiv.org/abs/1609.05537}{arXiv:1609.05537}}.

\bibitem{CPS96}
Xiao-Wen Chang, Christopher~C. Paige, and G.W. Stewart, \emph{New perturbation
  analyses for the {C}holesky factorization}, IMA Journal of Numerical Analysis
  \textbf{16} (1996), no.~4, 457--484.

\bibitem{CGLLTW19}
Nai-Hui Chia, Andr\'{a}s Gily\'{e}n, Tongyang Li, Han-Hsuan Lin, Ewin Tang, and
  Chunhao Wang, \emph{Sampling-based sublinear low-rank matrix arithmetic
  framework for dequantizing quantum machine learning}, Proceedings of the 52nd
  Annual ACM SIGACT Symposium on Theory of Computing (STOC 2020), p.~387–400,
  ACM, 2020, \mbox{\href{https://arxiv.org/abs/1910.06151}{arXiv:1910.06151}}.

\bibitem{chia2018quantum}
Nai-Hui Chia, Han-Hsuan Lin, and Chunhao Wang, \emph{Quantum-inspired sublinear
  classical algorithms for solving low-rank linear systems}, 2018,
  \mbox{\href{https://arxiv.org/abs/1811.04852}{arXiv:1811.04852}}.

\bibitem{cohen2018solving}
Michael~B. Cohen, Yin~Tat Lee, and Zhao Song, \emph{Solving linear programs in
  the current matrix multiplication time}, Proceedings of the 51st Annual ACM
  Symposium on Theory of Computing (STOC 2019), pp.~938--942, ACM, 2019,
  \mbox{\href{https://arxiv.org/abs/1810.07896}{arXiv:1810.07896}}.

\bibitem{DKM06}
P.~Drineas, R.~Kannan, and M.~Mahoney, \emph{Fast monte carlo algorithms for
  matrices i: Approximating matrix multiplication}, SIAM Journal on Computing
  \textbf{36} (2006), no.~1, 132--157.

\bibitem{FS94}
Shmuel Friedland and Wasin So, \emph{On the product of matrix exponentials},
  Linear Algebra and its Applications \textbf{196} (1994), 193 -- 205.

\bibitem{FKV04}
Alan Frieze, Ravi Kannan, and Santosh Vempala, \emph{Fast {M}onte-{C}arlo
  algorithms for finding low-rank approximations}, Journal of the {ACM}
  \textbf{51} (2004), no.~6, 1025--1041.

\bibitem{GH11}
Dan Garber and Elad Hazan, \emph{Approximating semidefinite programs in
  sublinear time}, Advances in Neural Information Processing Systems (NIPS
  2011), pp.~1080--1088, 2011.

\bibitem{GH12}
Dan Garber and Elad Hazan, \emph{Almost optimal sublinear time algorithm for
  semidefinite programming}, 2012,
  \mbox{\href{https://arxiv.org/abs/1208.5211}{arXiv:1208.5211}}.

\bibitem{GLT18}
Andr{\'a}s Gily{\'e}n, Seth Lloyd, and Ewin Tang, \emph{Quantum-inspired
  low-rank stochastic regression with logarithmic dependence on the dimension},
  2018, \mbox{\href{https://arxiv.org/abs/1811.04909}{arXiv:1811.04909}}.

\bibitem{grotschel1981ellipsoid}
Martin Gr{\"o}tschel, L{\'a}szl{\'o} Lov{\'a}sz, and Alexander Schrijver,
  \emph{The ellipsoid method and its consequences in combinatorial
  optimization}, Combinatorica \textbf{1} (1981), no.~2, 169--197.

\bibitem{gutoski2012parallel}
Gus Gutoski and Xiaodi Wu, \emph{Parallel approximation of min-max problems
  with applications to classical and quantum zero-sum games}, Proceedings of
  the 27th Annual IEEE Symposium on Computational Complexity (CCC 2012),
  pp.~21--31, IEEE, 2012.

\bibitem{Hazan}
Elad Hazan, \emph{Efficient algorithms for online convex optimization and their
  applications}, Ph.D. thesis, Princeton University, 2006.

\bibitem{jain2011parallel}
Rahul Jain and Penghui Yao, \emph{A parallel approximation algorithm for
  positive semidefinite programming}, Proceedings of the 52nd Annual IEEE
  Symposium on Foundations of Computer Science (FOCS 2011), pp.~463--471, IEEE,
  2011, \mbox{\href{https://arxiv.org/abs/1104.2502}{arXiv:1104.2502}}.

\bibitem{JLS19}
Dhawal Jethwani, Fran{\c{c}}ois~Le Gall, and Sanjay~K. Singh,
  \emph{Quantum-inspired classical algorithms for singular value
  transformation}, 2019,
  \mbox{\href{https://arxiv.org/abs/1910.05699}{arXiv:1910.05699}}.

\bibitem{jiang2020improved}
Haotian Jiang, Yin~Tat Lee, Zhao Song, and Sam Chiu-wai Wong, \emph{An improved
  cutting plane method for convex optimization, convex-concave games and its
  applications}, Proceedings of the 52nd Annual ACM SIGACT Symposium on Theory
  of Computing (STOC 2020), pp.~944--953, 2020,
  \mbox{\href{https://arxiv.org/abs/2004.04250}{arXiv:2004.04250}}.

\bibitem{jiang2020faster}
Shunhua Jiang, Zhao Song, Omri Weinstein, and Hengjie Zhang, \emph{Faster
  dynamic matrix inverse for faster {LP}s}, 2020,
  \mbox{\href{https://arxiv.org/abs/2004.07470}{arXiv:2004.07470}}.

\bibitem{kannan2017RandAlgNumLinAlg}
Ravindran Kannan and Santosh Vempala, \emph{Randomized algorithms in numerical
  linear algebra}, Acta Numerica \textbf{26} (2017), 95--135.

\bibitem{kerenidis2016recommendation}
Iordanis Kerenidis and Anupam Prakash, \emph{Quantum recommendation systems},
  Proceedings of the 8th Innovations in Theoretical Computer Science Conference
  (ITCS 2017), pp.~49:1--49:21, 2017,
  \mbox{\href{https://arxiv.org/abs/1603.08675}{arXiv:1603.08675}}.

\bibitem{kerenidis2018quantum}
Iordanis Kerenidis and Anupam Prakash, \emph{A quantum interior point method
  for {LP}s and {SDP}s}, 2018,
  \mbox{\href{https://arxiv.org/abs/1808.09266}{arXiv:1808.09266}}.

\bibitem{khachiyan1980polynomial}
Leonid~G. Khachiyan, \emph{Polynomial algorithms in linear programming}, USSR
  Computational Mathematics and Mathematical Physics \textbf{20} (1980), no.~1,
  51--68.

\bibitem{LRS15}
James~R. Lee, Prasad Raghavendra, and David Steurer, \emph{Lower bounds on the
  size of semidefinite programming relaxations}, Proceedings of the 47th Annual
  ACM Symposium on Theory of Computing (STOC 2015), ACM, 2015,
  \mbox{\href{https://arxiv.org/abs/1411.6317}{arXiv:1411.6317}}.

\bibitem{lee2015faster}
Yin~Tat Lee, Aaron Sidford, and Sam Chiu-wai Wong, \emph{A faster cutting plane
  method and its implications for combinatorial and convex optimization},
  Proceedings of the 56th Annual IEEE Symposium on Foundations of Computer
  Science (FOCS 2015), pp.~1049--1065, IEEE, 2015,
  \mbox{\href{https://arxiv.org/abs/1508.04874}{arXiv:1508.04874}}.

\bibitem{luby1993parallel}
Michael Luby and Noam Nisan, \emph{A parallel approximation algorithm for
  positive linear programming}, Proceedings of the 25th Annual ACM Symposium on
  Theory of Computing (STOC 1993), pp.~448--457, ACM, 1993.

\bibitem{mitchell2003polynomial}
John~E. Mitchell, \emph{Polynomial interior point cutting plane methods},
  Optimization Methods and Software \textbf{18} (2003), no.~5, 507--534.

\bibitem{nesterov1992conic}
Yurii Nesterov and Arkadi Nemirovsky, \emph{Conic formulation of a convex
  programming problem and duality}, Optimization Methods and Software
  \textbf{1} (1992), no.~2, 95--115.

\bibitem{PW09}
David Poulin and Pawel Wocjan, \emph{Sampling from the thermal quantum {G}ibbs
  state and evaluating partition functions with a quantum computer}, Physical
  Review Letters \textbf{103} (2009), no.~22, 220502,
  \mbox{\href{https://arxiv.org/abs/0905.2199}{arXiv:0905.2199}}.

\bibitem{preskill2018NISQ}
John Preskill, \emph{Quantum computing in the {NISQ} era and beyond}, Quantum
  \textbf{2} (2018), 79,
  \mbox{\href{https://arxiv.org/abs/1801.00862}{arXiv:1801.00862}}.

\bibitem{tang2018quantumv1}
Ewin Tang, \emph{A quantum-inspired classical algorithm for recommendation
  systems}, 2018,
  \mbox{\href{https://arxiv.org/abs/1807.04271v1}{arXiv:1807.04271v1}}, The
  first version online.

\bibitem{tang2018quantum2}
Ewin Tang, \emph{Quantum-inspired classical algorithms for principal component
  analysis and supervised clustering}, 2018,
  \mbox{\href{https://arxiv.org/abs/1811.00414}{arXiv:1811.00414}}.

\bibitem{tang2018quantum}
Ewin Tang, \emph{A quantum-inspired classical algorithm for recommendation
  systems}, Proceedings of the 51st Annual ACM SIGACT Symposium on Theory of
  Computing (STOC 2019), pp.~217--228, ACM, 2019,
  \mbox{\href{https://arxiv.org/abs/1807.04271}{arXiv:1807.04271}}.

\bibitem{vandenberghe1996semidefinite}
Lieven Vandenberghe and Stephen Boyd, \emph{Semidefinite programming}, SIAM
  Review \textbf{38} (1996), no.~1, 49--95.

\bibitem{Wu10}
Xiaodi Wu, \emph{Parallelized solution to semidefinite programmings in quantum
  complexity theory}, 2010,
  \mbox{\href{https://arxiv.org/abs/1009.2211}{arXiv:1009.2211}}.
\end{thebibliography}
\end{document}